\documentclass[11pt]{article}
\pdfoutput=1
\usepackage[centertags]{amsmath}
\usepackage[square, comma, sort&compress,numbers]{natbib}
\usepackage{array,multirow}
\numberwithin{equation}{section}
\usepackage{amssymb,amsfonts,amsthm}
\usepackage{graphicx}
\usepackage{color}
\usepackage{mathtools,bm}
\usepackage{accents}
\usepackage{soul}

\usepackage{epsfig}
\usepackage{bbold}

\usepackage{wrapfig}
\usepackage{float}
\usepackage{soul}
\usepackage{tkz-euclide}
\usepackage{braket}

\usepackage{tikz,pgf}
\usetikzlibrary{shapes}
\usetikzlibrary{calc}
\usetikzlibrary{decorations.pathmorphing}
\usetikzlibrary{decorations.pathreplacing,shapes.misc}
\usetikzlibrary{positioning}
\usetikzlibrary{arrows}
\usetikzlibrary{decorations.markings}
\usetikzlibrary{shadings}

\usetikzlibrary{intersections}

\def\be{\be}
\def\ee{\ee}
\def\ba{\begin{array}}
\def\ea{\end{array}}

\def\bst{\begin{split}}
\def\est{\end{split}}

\def\dps{\displaystyle}

\def\1{\tilde{1}}
\def\2{\tilde{2}}
\def\3{\tilde{3}}

\newdimen\tableauside\tableauside=1.0ex
\newdimen\tableaurule\tableaurule=0.4pt
\newdimen\tableaustep
\def\phantomhrule#1{\hbox{\vbox to0pt{\hrule height\tableaurule
width#1\vss}}}
\def\phantomvrule#1{\vbox{\hbox to0pt{\vrule width\tableaurule
height#1\hss}}}
\def\sqr{\vbox{%
\phantomhrule\tableaustep

\hbox{\phantomvrule\tableaustep\kern\tableaustep\phantomvrule\tableaustep}%
\hbox{\vbox{\phantomhrule\tableauside}\kern-\tableaurule}}}
\def\squares#1{\hbox{\count0=#1\noindent\loop\sqr
\advance\count0 by-1 \ifnum\count0>0\repeat}}
\def\tableau#1{\vcenter{\offinterlineskip
\tableaustep=\tableauside\advance\tableaustep by-\tableaurule
\kern\normallineskip\hbox
{\kern\normallineskip\vbox
{\gettableau#1 0 }%
\kern\normallineskip\kern\tableaurule}%
\kern\normallineskip\kern\tableaurule}}
\def\gettableau#1 {\ifnum#1=0\let\next=\null\else
\squares{#1}\let\next=\gettableau\fi\next}

\newtheorem{prop}{Proposition}[section]

\newcommand{\bref}[1]{\textbf{\ref{#1}}}

\newcommand{\RR}{\mathbb{R}}

\def\cD{\mathcal{D}}

\def\cF{\mathcal{F}}

\def\cH{\mathcal{H}}

\def\cM{\mathcal{M}}

\def\cR{\mathcal{R}}

\numberwithin{equation}{section} \makeatletter
\@addtoreset{equation}{section}

\def\be{\begin{equation}}
\def\ee{\end{equation}}
\def\ba{\begin{array}}
\def\ea{\end{array}}

\def\dps{\displaystyle}

\def\ba{\begin{array}}
\def\ea{\end{array}}

\def\dps{\displaystyle}

\def\rvac{|0\rangle}

\def\bx{{\bf x}}

\def\Li2{\operatorname{Li_2}}

\def\Tr{\operatorname{Tr}}
\def\C2{\text{C}_2}

\def\tcpw{{\tt t}CPW }
\def\tcpws{{\tt t}CPWs }
\def\tcft{{\tt t}CFT }

\def\dcpw{{\tt d}CPW }
\def\dcpws{{\tt d}CPWs}
\def\dcft{{\tt d}CFT }
\def\dcfts{{\tt d}CFTs }
\def\fcpw{{\tt f}CPW }
\def\fcpws{{\tt f}CPWs }
\def\fcft{{\tt f}CFT }

\def\tPsi{{}^{(t)}\Psi_{\Delta}^{\Delta_1,...,\Delta_4}(\bx)}
\def\tG{{}^{(t)}G_{\Delta}^{\Delta_1,...,\Delta_4}(\bx)}
\def\tGs{{}^{(t)}G_{\widetilde\Delta}^{\Delta_1,...,\Delta_4}(\bx)}

\def\tPsiDL{{}^{(t)}\Psi_{\Delta}^{\Delta,\Delta_\phi, \Delta,\widetilde\Delta_\phi}(x, qx,0,y)}

\newcommand{\dl}{\Delta}
\newcommand{\td}{\widetilde{\Delta}}

\tikzset{
    block/.style={
        draw=black,
        line width=0.9pt
    },
    snakeblock/.style={
        draw=black,
        line width=ultra thick
    },
    propagator/.style={
        draw=black,
        line width=2.6pt
    },
}

\usepackage{jheppub}
\makeatletter
\def\@fpheader{\vspace{-.1cm}}
\makeatother

\title{\centering{
    Thermal conformal partial waves\\
    from flat-space and defect CFT
    }}

\author{Konstantin Alkalaev}
\author{\;\;Semyon Mandrygin}
\author{\;\;Vladimir Samsonov}
\affiliation{I.E. Tamm Department of Theoretical Physics \\P.N. Lebedev Physical
Institute, 119991 Moscow, Russia}

\emailAdd{alkalaev@lpi.ru, semyon.mandrygin@gmail.com, samsonov.vs@phystech.edu}

\abstract{We establish a  correspondence between conformal partial waves on flat, thermal, and defect backgrounds using the shadow formalism. We demonstrate that scalar one-point thermal blocks can be systematically obtained from their four-point flat-space and  two-point  defect counterparts by considering specific operator configurations. This framework allows us to derive the thermal Casimir equation as a  diagonal reduction of the flat-space Casimir system without introducing chemical potentials. We further show that  defect two-point blocks with a spin-$l$ exchange  operator  correspond to thermal one-point blocks for an external spin-$l$ operator.}

\begin{document}
\maketitle

\flushbottom

\section{Introduction}

In a $d$-dimensional conformal field theory (CFT), the Hilbert space of states $\cH$ is organized into $o(d+1,1)$ Verma modules. Each module is generated by a primary operator together with its descendants. For correlation functions on flat space $\RR^d$, this representation structure is implemented through the operator product expansion (OPE), which organizes the correlators into sums over conformal families, with the kinematic dependence given by conformal blocks. The corresponding model-dependent expansion coefficients are constrained by OPE associativity, leading to the conformal bootstrap equations \cite{Belavin:1984vu,Cardy:1986ie,Rattazzi:2008pe}.

A natural generalization is obtained by placing CFT on more general backgrounds, including thermal manifolds $S^1_\beta \times \cM^{d-1}$ and $\RR^d$ in the presence of flat $p$-dimensional defects. Finite temperature and defects break the full conformal symmetry $o(d+1,1)$, leading to more intricate kinematics and introducing new conformal data. Together with the ordinary OPE coefficients, this data is subject to additional consistency constraints, often referred to as the thermal and defect bootstrap \cite{Iliesiu:2018fao,Liendo:2012hy}. However, these programs are less developed than the flat-space conformal bootstrap: standard numerical methods are not directly applicable, and a systematic understanding of the relevant conformal blocks is still lacking.

In this paper, we address the latter issue by analyzing conformal partial waves (CPWs) and the corresponding conformal blocks on flat, thermal, and defect backgrounds.\footnote{Flat-space, thermal, and defect CPWs will be denoted as \fcpw\hspace{-1.5mm}, \tcpw\hspace{-1.5mm}, and \dcpw\hspace{-1.5mm}, while the respective CFTs will be denoted as \fcft\hspace{-1.5mm}, \tcft\hspace{-1.5mm}, and \dcft\hspace{-1.5mm}.}
Our central claim is that these seemingly different settings are nevertheless closely related: thermal blocks can be systematically obtained from both flat-space and defect CFT constructions. This framework may therefore provide a new perspective on bootstrapping thermal and defect CFTs.

\paragraph{Thermal CFT.} Thermal conformal correlation functions are defined as traces over the same  Hilbert space $\mathcal{H}$, with states weighted by the dilatation operator $D$ \cite{El-Showk:2011yvt}. For a scalar primary operator $\phi(x)$ of conformal dimension $\Delta_\phi$, with $x \in \mathbb{R}^d$, the thermal one-point function is given by
\be
\label{th_corr_def}
\big\langle \phi(x) \big\rangle_{\beta} = \Tr_{\mathcal{H}} \Big[\,\phi(x)\, q^D  \,\Big]\,,
\qquad
q = e^{-\beta}\,,
\ee
where $\beta$ is the inverse temperature. The introduction of temperature  breaks the full conformal symmetry $o(d+1,1)$ down to the subalgebra $o(1,1)\oplus o(d)$ commuting with  $D$ \cite{Gobeil:2018fzy}. As a result,  thermal averages  are subject to fewer Ward identities compared to correlation functions on $\mathbb{R}^d$. Nevertheless, the residual symmetries are sufficient to completely fix the $x$-dependence of this observable, which remains a non-trivial function of the thermal parameter $q$.

Flat space  $\mathbb{R}^d$ is related to the cylinder $\mathbb{R}\times S^{d-1}$ through the Weyl map $|x| = e^{\tau}$. The relation $\phi_{\mathbb R\times S^{d-1}}(\tau,\Omega)=|x|^{\Delta_\phi}\phi(x)$ between primary operators on these two backgrounds induces the following transformation rule for the one-point function \eqref{th_corr_def}:
\be
\big\langle\phi_{\mathbb{R}\times S^{d-1}}(\tau,\Omega)\big\rangle_\beta
 = |x|^{\Delta_\phi}
\big\langle\phi(x)\big\rangle_\beta\,,
\qquad
\Omega \in S^{d-1} \subset \mathbb{R}^d\,.
\ee
Using $\phi_{\mathbb{R}\times S^{d-1}}(\tau,\Omega)=e^{\tau D}\phi_{\mathbb{R}\times S^{d-1}}(0,\Omega)e^{-\tau D}$ and the cyclicity of the trace, one can show that (see e.g. \cite{Alkalaev:2024jxh})\footnote{For higher-point thermal correlation functions, the analogous argument yields the KMS conditions, which serve as crossing-symmetry equations in the thermal bootstrap \cite{Iliesiu:2018fao, El-Showk:2011yvt}.}
\be
\label{th_corr_periodic}
\big\langle\phi_{\mathbb{R}\times S^{d-1}}(\tau+\beta,\Omega)\big\rangle_\beta
=
\big\langle\phi_{\mathbb{R}\times S^{d-1}}(\tau,\Omega)\big\rangle_\beta\,.
\ee
Thus, in cylindrical coordinates, the one-point function is periodic in Euclidean time and, therefore, defines a thermal observable on $S^1_\beta\times S^{d-1}$. Moreover, the residual Ward identities imply that the one-point function is not merely periodic in $\tau$, but is independent of both $\Omega$ and $\tau$.

In the high-temperature limit $\beta \to 0$, the geometry $S^1_\beta\times S^{d-1}$ effectively reduces to $S^1_\beta \times \mathbb{R}^{d-1}$, where the one-point function is fixed up to a model-dependent constant \cite{El-Showk:2011yvt}:
\be
\big\langle \phi(\tau, \vec{x}) \big\rangle_{S^1_\beta \times \mathbb{R}^{d-1}} = \frac{b_\phi}{\beta^{\Delta_\phi}}\,,
\qquad
\vec{x} \in \mathbb{R}^{d-1}\,.
\ee
Moreover, two-point functions on $S^1_\beta \times \mathbb{R}^{d-1}$ admit an explicit conformal block decomposition in terms of Gegenbauer polynomials \cite{Iliesiu:2018fao}. This well-understood structure permits a detailed analysis of the KMS conditions and has driven recent progress in the thermal conformal bootstrap \cite{Petkou:2018ynm, Marchetto:2023xap, Barrat:2025wbi, Barrat:2025nvu, Barrat:2025twb, Niarchos:2025cdg, Buric:2025anb}.

In contrast, for thermal correlators on $S^1_\beta \times S^{d-1}$, a systematic understanding is much less developed, despite important recent progress \cite{Gobeil:2018fzy, Kraus:2017ezw, Alkalaev:2023evp, Alkalaev:2024jxh, David:2024pir, David:2025tqn, Ammon:2025cdz, Buric:2024kxo, Buric:2025uqt, Anand:2025mfh, Buric:2025fye}.\footnote{Related directions include searching  for higher-dimensional analogues of modular invariance \cite{Cardy:1991kr, Lei:2024oij, Allameh:2024qqp, Shaghoulian:2016gol, DiPietro:2014bca}, as well as the study of CFTs on higher genus manifolds \cite{Benjamin:2023qsc, Simmons-Duffin:2025qox}.}
The main difficulty is that even the one-point function \eqref{th_corr_def} admits a non-trivial decomposition into conformal blocks \cite{El-Showk:2011yvt, Gobeil:2018fzy}:
\be
\label{th_corr_CB_expansion}
\big\langle \phi(x) \big\rangle_{\beta} = \sum_{\Delta} C_{\Delta \Delta_\phi \Delta}\, \cF_{\Delta}^{\Delta_\phi}(q, x)  + \text{spinning contributions}\,,
\ee
where $C_{\Delta \Delta_\phi \Delta}$ are the OPE coefficients. The scalar thermal conformal block $\cF_{\Delta}^{\Delta_\phi}(q, x)$ can be represented as an infinite sum over matrix elements, reflecting the underlying trace over  $\mathcal{H}$. This leads to an asymptotic expansion at  small $q$:
\be
\label{tF_asymptot}
\cF_{\Delta}^{\Delta_\phi}(q, x) = |x|^{-\Delta_\phi}\, q^{\Delta} \big( 1 + O(q)\big)\,,
\ee
where the $x$-dependence is fixed by the residual Ward identities.
To compute these blocks without evaluating the trace directly, one can employ auxiliary methods, such as the AdS integral representation \cite{Gobeil:2018fzy} or the shadow formalism \cite{Alkalaev:2024jxh}. These approaches yield closed-form expressions in terms of generalized hypergeometric functions, thereby making the thermal  blocks amenable for further analytic manipulations.


\paragraph{Thermal blocks from flat-space and defect CFTs.}
Several developments suggest that thermal observables may admit a description in terms of more standard data of a different, and often more controllable, underlying theory. In particular, motivated by the Eigenstate Thermalization Hypothesis \cite{Deutsch:1991, Srednicki:1994}, it has been proposed  that correlation functions evaluated in heavy states reproduce thermal expectation values in CFTs with an effective temperature determined  by the conformal dimension of the corresponding heavy primary operator \cite{Asplund:2014coa, Fitzpatrick:2015qma, Lashkari:2016vgj,Fitzpatrick:2015zha}. More recently, thermal behavior characteristic of black holes was shown to be encoded in states of integrable spin chains \cite{Kristjansen:2025xqo}. Furthermore, it has been observed that certain four-point conformal integrals admit representations in terms of thermal one-point functions on $S_\beta^1 \times \mathbb{R}^{d-1}$ \cite{Petkou:2021zhg, Karydas:2023ufs, Karydas:2025tfs}, while a reduction of four-point parametric conformal integrals yields one-point thermal conformal partial waves on $S_\beta^1 \times S^{d-1}$ \cite{Alkalaev:2024jxh}.

Although these mechanisms are different, they share a common perspective: thermal behavior emerges from an underlying zero-temperature description once distinguished conditions are imposed on dynamical or kinematical data. In other words, thermal correlators should not be viewed exclusively as independent observables, but rather as arising from more general data upon restriction to a specific regime. In the present work, we aim to make this relation precise at the level of conformal partial waves and conformal blocks.

\vspace{2mm}
\noindent {\it Flat-space CFT.} We claim that  one-point scalar thermal CPWs on $S_\beta^1 \times S^{d-1}$ arise from  the  diagonal limit of ordinary four-point CPWs in flat space $\mathbb{R}^d$. Namely, \tcpws with  external and intermediate dimensions $\Delta_\phi$ and  $\Delta$ are obtained from $t$-channel \fcpws with  intermediate dimension $\Delta$ and external dimensions
\be
\label{conformal_dimensions}
\Delta_1 =  \Delta\,,
\qquad
\Delta_2 = \Delta_\phi\,,
\qquad
\Delta_3 = \Delta\,,
\qquad
\Delta_4 = d - \Delta_\phi \equiv \widetilde \Delta_\phi\,,
\ee
evaluated in the  diagonal configuration,
\be
\label{diag1}
x_2 = q\, x_1\,,
\qquad
x_3 = 0\,,
\qquad
x_4 \to \infty\,,
\ee
see fig. \bref{fig:diag}. This specific placement of operator insertions is controlled by $q \in (0,1)$, which is identified with the thermal parameter.\footnote{It is worth noting the parallel with \cite{Petkou:2021zhg, Karydas:2023ufs, Karydas:2025tfs}, where temperature and chemical potential arise through a parametrization of cross-ratios.} Such a  correspondence between CPWs extends directly to the respective conformal blocks and their shadows, providing a one-to-one identification.

\begin{figure}
\centering
\begin{tikzpicture}[scale=1]

\draw[propagator] (0,-0.76) -- (0,0.76);

\draw[block] (0,0.74) -- (-1.0,1.64);
\draw[block] (0,0.74) -- (1.0,1.64);

\draw[block] (0,-0.74) -- (-1.0,-1.64);
\draw[block] (0,-0.74) -- (1.0,-1.64);

\fill (-1.0,1.64) circle (1.2pt);
\fill (1.0,1.64) circle (1.2pt);
\fill (-1.0,-1.64) circle (1.2pt);
\fill (1.0,-1.64) circle (1.2pt);

\node[left=2pt]  at (-1.0,1.56) {$x_1$};
\node[right=2pt] at (1.0,1.56) {$\infty$};

\node[left=2pt]  at (-1.0,-1.56) {$q x_1$};
\node[right=2pt] at (1.0,-1.56) {$0$};;

\node at (-1.0,1.13) {$\Delta$};
\node at (1.0,1.13) {$\widetilde{\Delta}_\phi$};

\node at (-1.0,-1.13) {$\Delta_\phi$};
\node at (1.0,-1.13) {$\Delta$};

\node[right=3pt] at (0,0) {$\Delta$};

\draw[
    -latex,
    line width=3.0pt,
    cyan!55!blue
] (2.0,0) -- (3.5,0);

\draw[propagator] (5.0,0.84) circle (0.8);

\draw[block] (5.0,0.04) -- (5.0,-1.56);

\fill (5.0,-1.56) circle (1.2pt);

\node[right=3pt] at (5.0,-1.56) {$x_1$};

\node[right=3pt] at (5.8,0.84) {$\Delta$};
\node[right=3pt] at (5.0,-0.76) {$\Delta_\phi$};

\end{tikzpicture}
\caption{Correspondence between the four-point CPW in the $t$-channel for the operator configuration \eqref{conformal_dimensions}--\eqref{diag1} and the one-point scalar thermal CPW.}
\label{fig:diag}
\end{figure}
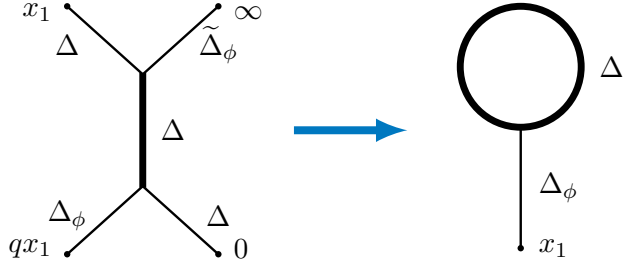

\begin{figure}
\centering
\begin{tikzpicture}[scale=1]

\draw[propagator] (0,-1.46) -- (0,0.498);

\draw[block] (0,0.476) -- (-1.22,1.64);
\draw[block] (0,0.476) -- (1.22,1.64);

\node[
    star,
    star points=6,
    fill=red!75!,
    star point ratio=0.4,
    minimum size=0.198cm
] at (0,-1.46) {};

\fill (-1.22,1.64) circle (1.2pt);
\fill (1.22,1.64) circle (1.2pt);

\node[left=2pt]  at (-1.22,1.56) {$x_1$};
\node[right=2pt] at (1.22,1.56) {$q x_1$};

\node at (-1.10,1.05) {$\Delta$};
\node at (1.10,1.05) {$\widetilde{\Delta}$};

\node[right=3pt] at (0,-0.76) {$\Delta_\phi, l$};

\draw[
    -latex,
    line width=3.0pt,
    cyan!55!blue
] (2.0,0.04) -- (3.5,0.04);

\draw[propagator] (5.0,0.84) circle (0.8);

\draw[block] (5.0,0.04) -- (5.0,-1.56);

\fill (5.0,-1.56) circle (1.2pt);

\node[right=3pt] at (5.0,-1.56) {$x_1$};

\node[right=3pt] at (5.8,0.84) {$\Delta$};
\node[right=3pt] at (5.0,-0.76) {$\Delta_\phi, l$};

\end{tikzpicture}
\caption{Correspondence between the two-point defect CPW in the bulk channel for the operator configuration \eqref{conformal_dimensions_defect}--\eqref{diag2} (a point-like defect is represented by the red star) and the one-point spinning thermal CPW.}
\label{fig:defect}
\end{figure}
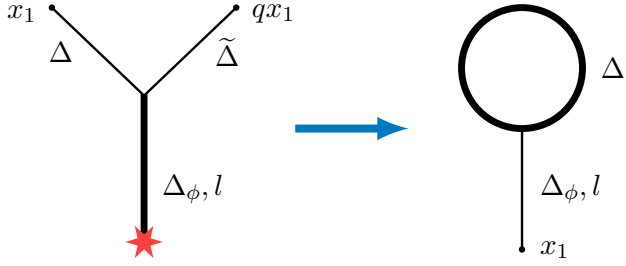

\vspace{2mm}

\noindent {\it Defect CFT.} It turns out that within the present framework introducing spin does not yield a flat-space CPW of one type or another. Rather, this generalization is naturally realized within the defect CFT.  We show that the one-point \tcpw with a spin-$l$ external operator and a scalar intermediate operator can be reformulated as a bulk-channel two-point \dcpw in the presence of a point-like defect. A characteristic feature of this correspondence is that the roles of  the external and intermediate  operators are interchanged. Namely, we consider a bulk-channel \dcpw  for two scalar external operators of conformal dimensions $\Delta_1$ and $\Delta_2$, and a spin-$l$ intermediate operator of conformal dimension $\Delta_k$. These conformal dimensions and positions of operators are related  to those  in the one-point \tcpw for an external spin-$l$ operator as follows
\be
\label{conformal_dimensions_defect}
\Delta_1=\Delta,\qquad
\Delta_2=d-\Delta \equiv \widetilde \Delta ,\qquad
\Delta_k=\Delta_\phi\,,
\ee
\be
\label{diag2}
x_2 = q\, x_1\,,
\ee
see fig. \bref{fig:defect}. There is no direct identification between individual thermal and bulk conformal blocks. Instead, the thermal block receives contributions from both the bulk conformal block and its shadow combined to satisfy the thermal asymptotic  condition \eqref{tF_asymptot}. This differs from the flat-space case discussed above.

An auxiliary observation follows from comparing fig. \bref{fig:diag} and  fig. \bref{fig:defect}. The two-point \dcpw with dimensions \eqref{conformal_dimensions_defect} is related to the four-point \fcpw with dimensions \eqref{conformal_dimensions}, see fig. \bref{fig:defect-flat}. Importantly, this correspondence holds beyond the diagonal limit: in the flat-space configuration only two points are fixed, $x_3=0$ and $x_4 \to \infty$, while the remaining two points are kept arbitrary and become the positions of the bulk operators in the defect setup.

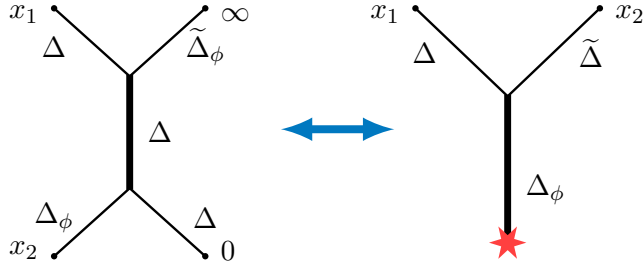
\begin{figure}
\centering
\begin{tikzpicture}[scale=1]

\draw[propagator] (5,-1.46) -- (5,0.496);

\draw[block] (5,0.476) -- (5-1.22,1.64);
\draw[block] (5,0.476) -- (5+1.22,1.64);

\node[
    star,
    star points=6,
    fill=red!75!,
    star point ratio=0.4,
    minimum size=0.198cm
] at (5,-1.46) {};

\fill (5-1.22,1.64) circle (1.2pt);
\fill (5+1.22,1.64) circle (1.2pt);

\node[left=2pt]  at (5-1.22,1.56) {$x_1$};
\node[right=2pt] at (5+1.22,1.56) {$x_2$};

\node at (5-1.10,1.05) {$\Delta$};
\node at (5+1.10,1.05) {$\widetilde{\Delta}$};

\node[right=3pt] at (5,-0.76) {$\Delta_\phi$};

\draw[
    latex-latex,
    line width=3.0pt,
    cyan!55!blue
] (2.0,0.04) -- (3.5,0.04);

\draw[propagator] (0,-0.76) -- (0,0.76);

\draw[block] (0,0.74) -- (-1.0,1.64);
\draw[block] (0,0.74) -- (1.0,1.64);

\draw[block] (0,-0.74) -- (-1.0,-1.64);
\draw[block] (0,-0.74) -- (1.0,-1.64);

\fill (-1.0,1.64) circle (1.2pt);
\fill (1.0,1.64) circle (1.2pt);
\fill (-1.0,-1.64) circle (1.2pt);
\fill (1.0,-1.64) circle (1.2pt);

\node[left=2pt]  at (-1.0,1.56) {$x_1$};
\node[right=2pt] at (1.0,1.56) {$\infty$};

\node[left=2pt]  at (-1.0,-1.56) {$x_2$};
\node[right=2pt] at (1.0,-1.56) {$0$};

\node at (-1.0,1.13) {$\Delta$};
\node at (1.0,1.13) {$\widetilde{\Delta}_\phi$};

\node at (-1.0,-1.13) {$\Delta_\phi$};
\node at (1.0,-1.13) {$\Delta$};

\node[right=3pt] at (0,0) {$\Delta$};

\end{tikzpicture}
\caption{Correspondence between the four-point CPW in the $t$-channel for the operator configuration \eqref{conformal_dimensions}--\eqref{diag1} and the two-point defect CPW in the bulk channel for the operator configuration \eqref{conformal_dimensions_defect}--\eqref{diag2} (a point-like defect is represented by the red star).}
\label{fig:defect-flat}
\end{figure}

\paragraph{Thermal Casimir equations.} The correspondence between CPWs provides a way to circumvent some of the difficulties that are not readily tractable by direct methods in thermal CFT. More specifically, thermal conformal blocks do not inherit many of the properties of their flat-space counterparts that are essential for the standard bootstrap analysis. This reflects the reduced conformal symmetry of the thermal background, becoming particularly evident at the level of the Casimir equations. In flat space, conformal blocks are fixed as the contributions of individual irreducible representations of the conformal group $o(d+1,1)$ to conformal correlation functions. Consequently, the Casimir equations arise directly from the action of the Casimir operators on these representations \cite{Dolan:2003hv,Dolan:2011dv}.

At finite temperature, the thermal conformal block still encodes the contribution of an irreducible representation to a correlation function. However, the thermal correlator is no longer a vacuum expectation value, but rather a trace over the full space of states. It prevents one from deriving the corresponding differential equations directly from the action of the Casimir operator. To overcome this difficulty, one can introduce chemical potentials, which serve to track the conserved charges within the trace and restore a structure that allows one to formulate the corresponding Casimir equations \cite{Gobeil:2018fzy}.\footnote{This also applies to large-$c$ torus CFT$_2$ and the respective global conformal blocks, where the complex-valued modular parameter $q$ simultaneously  encodes both the temperature and  the chemical potential, see \cite{Kraus:2017ezw,Alkalaev:2018qaz,Alkalaev:2022kal}.} This extended framework has recently enabled practical recursive methods for computing the coefficients of thermal conformal block in the presence of chemical potentials for both scalar and spinning operators \cite{Buric:2024kxo, Buric:2025uqt}.

The existence of two distinct methods for relating thermal CPWs to other CFT configurations provides an immediate advantage. Namely, we show that the thermal Casimir equations  can be obtained by taking the diagonal limit of the ordinary flat-space Casimir system evaluated for the specific operator configuration \eqref{conformal_dimensions}. Remarkably, this reduction procedure is already well established in the literature \cite{Hogervorst:2013kva}. The correspondence between CPWs allows us to bypass the introduction of chemical potentials and turn the diagonal limit into a direct method for deriving thermal Casimir equations.

Alternatively, since the derivation of the Casimir equations in defect CFT proceeds in close analogy with the flat-space case \cite{Billo:2016cpy, Gadde:2016fbj}, this correspondence could, in principle, lay the groundwork for obtaining the Casimir equations for spinning thermal blocks as diagonal limits of the corresponding Casimir systems in the presence of defects.

\paragraph{Organization of the paper.} In section \bref{sec:flat} we introduce  the shadow formalism, define  four-point scalar \fcpws in the $t$-channel, examine the underlying  conformal integrals, and discuss their diagonal limit. Section \bref{sec:thermal} is devoted to the construction of \tcpws  and the derivation of the scalar thermal conformal block. We demonstrate  how the one-point \tcpw can be related to the four-point \fcpw by considering  a specific operator  configuration.  In section \bref{sec:defect_scalar} we reformulate the construction within the defect CFT setup and clarify the relation between thermal and defect conformal blocks. The spinning extension is considered in section \bref{sec:spin}. In section \bref{sec:casimir} we derive the thermal Casimir equations from the diagonal reduction of the flat-space Casimir system. We conclude in section \bref{sec:conclusion} with a  summary of  our results and discuss  future perspectives. Appendix \bref{app:functions} collects necessary information on hypergeometric  functions and presents the derivation of the reduction formula, which relates the Appell $F_4$ function to the ${}_3 F_2$ hypergeometric function. Appendix \bref{app:self-dualtCPW} details  the regularization procedure used to calculate \tcpws in the case of self-dual conformal dimensions, $\Delta = \tilde \Delta$.

\section{Flat-space conformal partial waves}
\label{sec:flat}

The shadow formalism provides a convenient method for computing conformal blocks in $d$-dimensional CFT \cite{Ferrara:1972uq, Ferrara:1972ay, Ferrara:1972xe, Ferrara:1972kab, Dobrev:1977qv, Fradkin:1978pp, Fradkin:1997df, Dolan:2011dv, Simmons-Duffin:2012juh}.  Its main advantage lies in representing CPWs as integrals over products of three-point functions, thereby avoiding a direct summation over descendant states. The conformal block can then be extracted from a given CPW by isolating the physical contribution with dimension $\Delta$ and discarding the unphysical shadow contribution with dimension $\widetilde{\Delta} = d-\Delta$. We review this construction in a form tailored to the correspondence developed below, starting by introducing the conformally invariant operator
\be
\label{projector_def}
\Pi_{\Delta}
=
\int_{\mathbb{R}^{d}} {\rm d}^d x_0\,
\mathcal{O}(x_0)\ket{0}\bra{0}\widetilde{\mathcal{O}}(x_0)\,,
\ee
which projects onto the conformal families of an (intermediate) operator $\mathcal{O}$ and its shadow $\widetilde{\mathcal{O}}$ with dual conformal dimensions  $\Delta$ and $\widetilde{\Delta}$, respectively. These operators are related through an integral transformation \cite{Simmons-Duffin:2012juh, Karateev:2017jgd, Karateev:2018oml}, the normalization can be chosen such that   \eqref{projector_def} is idempotent: $\Pi_{\Delta}\Pi_{\Delta'}=\delta_{\Delta\Delta'}\,\Pi_{\Delta}$.

Let us consider the  four-point correlation function of scalar primary operators $\phi_{\Delta_i}(x_i)$, with $x_i\in\mathbb{R}^d$, $i = 1,...,4$. Inserting the projector \eqref{projector_def} between pairs of these  operators introduces a ${\tt f}$CPW in the corresponding  exchange channels. In what follows, we consider the $t$-channel
\be
\label{CPW_projector}
\big\langle \phi_{\Delta_4}(x_4)\phi_{\Delta_1}(x_1)\,\Pi_{\Delta}\,\phi_{\Delta_2}(x_2)\phi_{\Delta_3}(x_3)\big\rangle
=
C_{\Delta_4\Delta_1\Delta}\,
C_{\widetilde{\Delta}\Delta_3\Delta_2}\,
{}^{(t)}\Psi_{\Delta}^{\Delta_1,...,\Delta_4}(\bx)\,,
\ee
where $\bx =\{x_1,...,x_4\}$ and  $C_{\Delta_a\Delta_b\Delta_c}$ are the OPE coefficients.\footnote{\label{fn:shadow_ratio} Note that the relation between $\mathcal{O}$ and  $\widetilde{\mathcal{O}}$ induces a corresponding relation between their structure constants, expressed in terms of $\Gamma$-functions of the conformal dimensions, see e.g. \cite{Karateev:2018oml,Alkalaev:2024jxh}. Here, the associated prefactor is not included in the definition of the ${\tt f}$CPW.}
From \eqref{projector_def} one finds that the $t$-channel \fcpw  admits the integral representation:
\be
\label{CPW_integral_V}
\tPsi
=
\int_{\mathbb{R}^d} {\rm d}^d x_0\,
V_{\Delta_4\Delta_1\Delta}(x_4,x_1,x_0)\,
V_{\widetilde{\Delta}\Delta_2\Delta_3}(x_0,x_2,x_3)\,,
\ee
where the kinematical part of the scalar three-point function is given by
\be
\label{corr_V}
V_{\Delta_i\Delta_j\Delta_k}(x_i,x_j,x_k)
=
X_{ij}^{\frac{\Delta_k-\Delta_i-\Delta_j}{2}}
X_{ik}^{\frac{\Delta_j-\Delta_i-\Delta_k}{2}}
X_{jk}^{\frac{\Delta_i-\Delta_j-\Delta_k}{2}}\,,
\qquad
X_{ij} = (x_i - x_j)^2\,.
\ee
Substituting \eqref{corr_V} into \eqref{CPW_integral_V} one obtains
\be
\label{CPW_integral}
\tPsi=
X_{14}^{\frac{\Delta-\Delta_1-\Delta_4}{2}}\,
X_{23}^{\frac{\widetilde{\Delta}-\Delta_2-\Delta_3}{2}}
\int_{\mathbb{R}^d} {\rm d}^d x_0\,
{X_{01}^{-b_1}X_{02}^{-b_2}X_{03}^{-b_3}X_{04}^{-b_4}}\,,
\ee
where the parameters are given by
\be
\ba{c}
\label{CPW_parameters}
\dps
b_1=\frac{\Delta+\Delta_{14}}{2}\,,
\qquad
b_2=\frac{\widetilde{\Delta}+\Delta_{23}}{2}\,,
\qquad
b_3=\frac{\widetilde{\Delta}-\Delta_{23}}{2}\,,
\qquad
b_4=\frac{\Delta-\Delta_{14}}{2}\,,
\\[12pt]
\dps
b_1+b_2+b_3+b_4=d\,,
\ea
\ee
with $\Delta_{ij}=\Delta_i-\Delta_j$.

\subsection{Conformal integrals}
\label{sec:conf_int}

CPWs are  naturally represented in terms of the  four-point conformal integral. This parametric integral\footnote{For a retrospective discussion of  conformal integrals, see \cite{Dolan:2000uw}.}
\be
\label{box_def}
I_4^{a_1,a_2,a_3,a_4}(\bx)
=
\int_{\mathbb{R}^d} \frac{{\rm d}^d x_0}{\pi^{\frac{d}{2}}}
\prod_{j=1}^4 X_{0j}^{-a_j},
\qquad
\sum_{j=1}^4 a_j=d,
\ee
admits a decomposition into a sum of four {\it basis} functions, which are conveniently enumerated by ordered  triples of indices \cite{Alkalaev:2025fgn, Alkalaev:2025zhg}:
\be
\label{box_sum}
I_4^{\bm a}(\bx)
=
\Phi_4^{\langle 234\rangle}(\bm a|\bx)
+\Phi_4^{\langle 134\rangle}(\bm a|\bx)
+\Phi_4^{\langle 124\rangle}(\bm a|\bx)
+\Phi_4^{\langle 123\rangle}(\bm a|\bx)\,,
\ee
where $\bm a = \{a_1,...,a_4\}$.
These functions are related by the action of the cyclic group $\mathbb Z_4=\{e,C_4,(C_4)^2,(C_4)^3\}$:
\be
\label{box_basis_1}
\ba{l}
\dps
\Phi_4^{\langle 234\rangle}(\bm a|\bx)=(C_4)^1\circ\Phi_4^{\langle 123\rangle}(\bm a|\bx)\,,
\\[10pt]
\dps
\Phi_4^{\langle 134\rangle}(\bm a|\bx)=(C_4)^2\circ\Phi_4^{\langle 123\rangle}(\bm a|\bx)\,,
\\[10pt]
\dps
\Phi_4^{\langle 124\rangle}(\bm a|\bx)=(C_4)^3\circ\Phi_4^{\langle 123\rangle}(\bm a|\bx)\,,
\ea
\ee
where the cyclic permutation $C_4=(1234)$ acts  on $\bm a$ and $\bx$ as $(a_i, x_i) \to (a_{i+1}, x_{i+1})$. As a consequence, the full conformal integral is generated from a single basis function represented in terms of the fourth Appell function  \eqref{F4} as
\be
\label{box_basis_2}
{
\Phi_4^{\langle 123\rangle}(\bm a|\bx)
=
\frac{X_{12}^{a_{34}}
X_{23}^{a_{14}}}{X_{24}^{a_4}X_{13}^{\frac{d}{2}-a_2}}\;\,
\Gamma\left[
\begin{array}{c}
a_{12},\ a_{23},\ \frac{d}{2}-a_2
\\[1mm]
a_1,\ a_2,\ a_3
\end{array}
\right]
F_4
\Bigg[
\begin{array}{c}
\frac{d}{2}-a_2,\ a_4
\\[1mm]
1-a_{12},\ 1-a_{23}
\end{array}
\Bigg|\,u,v
\Bigg],
}
\ee
where the $\Gamma$-prefactor is defined in \eqref{gammas} and $a_{ij}=a_i+a_j-\frac{d}{2}$. The arguments $u$ and $v$ are the cross-ratios:
\be
\label{uv}
u=\frac{X_{12}X_{34}}{X_{13}X_{24}}\,,
\qquad
v=\frac{X_{14}X_{23}}{X_{13}X_{24}}\,.
\ee
Under the cyclic permutation, they are interchanged as
$C_4\circ u=v$ and  $C_4\circ v=u$.

\subsection{Conformal blocks}
\label{sec:conf_block}
The four basis functions naturally organize into pairs according to their asymptotic behavior in a given OPE limit. Consequently, after substituting \eqref{box_sum} into \eqref{CPW_integral}, the \fcpw decomposes into a conformal block and a shadow block, each expressed as a sum of two basis functions.

More precisely, the OPE fixes the leading behavior of a conformal block in a given channel as two points approach each other. This allows one to identify which pair of basis functions contributes to the conformal block and which contributes to its shadow. From this perspective, the decomposition into conformal and shadow blocks is uniquely determined by the asymptotic properties of the basis functions: the four basis functions can be grouped into pairs in three different ways  corresponding to the three OPE channels.

We consider the $t$-channel, which is defined by the OPE in the pairs $\phi_2\phi_3$ and $\phi_1\phi_4$. The basis functions  are then  split  according to their asymptotics as $X_{23}\to 0$. Following this prescription, one obtains
\be
\label{CPW_sum}
\tPsi
=
K_{\widetilde{\Delta}}^{\Delta_3\Delta_2}\,\tG
+
K_{\Delta}^{\Delta_1\Delta_4}\,\tGs\,,
\ee
where
\be
\label{K_coefficient}
K_{\Delta}^{\Delta_i\Delta_j}
=
\pi^{\frac{d}{2}}
\Gamma\left[
\begin{array}{c}
\Delta-\frac{d}{2},\
\frac{\widetilde{\Delta}-\Delta_{ij}}{2},\
\frac{\widetilde{\Delta}+\Delta_{ij}}{2}
\\[1mm]
\widetilde{\Delta},\
\frac{\Delta-\Delta_{ij}}{2},\
\frac{\Delta+\Delta_{ij}}{2}
\end{array}
\right]
,
\qquad
\Delta_{ij}=\Delta_i-\Delta_j\,.
\ee
The $t$-channel conformal block and its shadow are  given by
\be
\label{block_sum}
\tG =
\Big(K_{\widetilde{\Delta}}^{\Delta_3\Delta_2}\Big)^{-1}\,X_{23}^{\frac{\widetilde{\Delta}-\Delta_2-\Delta_3}{2}}
X_{14}^{\frac{\Delta-\Delta_1-\Delta_4}{2}}
\Big(
\Phi_4^{\langle 234\rangle}(\bm b|\bx)
+
\Phi_4^{\langle 123\rangle}(\bm b|\bx)
\Big),
\ee
\be
\label{block_sum_s}
\tGs =
\Big(K_{\Delta}^{\Delta_1\Delta_4}\Big)^{-1}\,X_{23}^{\frac{\widetilde{\Delta}-\Delta_2-\Delta_3}{2}}
X_{14}^{\frac{\Delta-\Delta_1-\Delta_4}{2}}
\Big(
\Phi_4^{\langle 134\rangle}(\bm b|\bx)
+
\Phi_4^{\langle 124\rangle}(\bm b|\bx)
\Big),
\ee
where the parameters $\bm b = \{b_1,...,b_4\}$ are defined in \eqref{CPW_parameters}. The shadow block \eqref{block_sum_s} can be obtained from the conformal block \eqref{block_sum} via the substitution $\Delta\to\widetilde{\Delta}$. These relations manifest the splitting of the conformal integral into pairs of basis functions in the corresponding exchange channel. Selecting a different exchange channel and separating the corresponding conformal and shadow contributions amount to regrouping the same basis functions according to their asymptotic behavior. It is convenient to factor out the conformally covariant prefactor
\be
\label{block}
\tG
=
\frac{1}{X_{23}^{\frac{\Delta_2+\Delta_3}{2}}X_{14}^{\frac{\Delta_1+\Delta_4}{2}}}
\left(\frac{X_{24}}{X_{34}}\right)^{\frac{\Delta_{32}}{2}}
\left(\frac{X_{34}}{X_{13}}\right)^{\frac{\Delta_{14}}{2}}
\,{}^{(t)}g_{\Delta}^{\Delta_1,...,\Delta_4}(u,v)\,,
\ee
thereby  defining the dimensionless $t$-channel conformal block
\be
\label{block_bare}
\begin{split}
{}^{(t)}g_{\Delta}^{\Delta_1,...,\Delta_4}(u,v)
=
v^{\frac{\Delta}{2}}\,
\Gamma\left[
\begin{array}{c}
\frac{\Delta_{32}-\Delta_{14}}{2},\ \Delta
\\[1mm]
\frac{\Delta+\Delta_{32}}{2},\ \frac{\Delta-\Delta_{14}}{2}
\end{array}
\right]
&
F_4
\left[
\begin{array}{c}
\frac{\Delta-\Delta_{32}}{2},\ \frac{\Delta+\Delta_{14}}{2}
\\[1mm]
1+\frac{\Delta_{14}-\Delta_{32}}{2},\ 1-\frac{d}{2}+\Delta
\end{array}
\Bigg|\,u,v
\right]
\\[3mm]
+
v^{\frac{\Delta}{2}}u^{\frac{\Delta_{32}-\Delta_{14}}{2}}\,
\Gamma\left[
\begin{array}{c}
\frac{\Delta_{14}-\Delta_{32}}{2},\ \Delta
\\[1mm]
\frac{\Delta-\Delta_{32}}{2},\ \frac{\Delta+\Delta_{14}}{2}
\end{array}
\right]
&
F_4
\left[
\begin{array}{c}
\frac{\Delta+\Delta_{32}}{2},\ \frac{\Delta-\Delta_{14}}{2}
\\[1mm]
1+\frac{\Delta_{32}-\Delta_{14}}{2},\ 1-\frac{d}{2}+\Delta
\end{array}
\Bigg|\,u,v
\right].
\end{split}
\ee

In the $t$-channel  OPE limit $X_{23}\to 0$, the cross-ratios behave as $v\to 0$ and $u\to 1$, which lies outside the convergence domain $\sqrt{|u|}+\sqrt{|v|}<1$ of the Appell $F_4$ series. While this typically requires an appropriate analytic continuation (as discussed, for instance, in \cite{Dolan:2000ut}), for our purposes it is more convenient to retain the $(u,v)$-symmetric representation \eqref{block_bare}. This form serves as the natural starting point for taking  the diagonal limit.

\paragraph{Diagonal limit.}
Consider the parametrization of the cross-ratios $u$ and $v$ \cite{Dolan:2000ut}:
\be
  \label{uv_q}
  u = (1-q)(1- \bar q)\,,
  \qquad
  v = q \bar q\,,
\ee
where $q$ and $\bar{q}$ are complex conjugates.  The parametrization is chosen such that $q\to0$ in the $t$-channel OPE limit. The diagonal limit is defined by restricting to the configuration of operators in $\RR^d$ for which\footnote{The diagonal limit is a practical tool in the numerical conformal bootstrap, see e.g. \cite{El-Showk:2012cjh, Hogervorst:2013kva}.}
\be
  \label{diagonal_defenition}
   q = \bar q\,.
\ee
In the $(u,v)$-plane, this condition corresponds to the boundary of the convergence domain given by  $u = (1-\sqrt{v})^2$. In this regime, conformal blocks become functions of a single real variable $q$.\footnote{In the correspondence developed below, $q$ is identified as the thermal parameter. The diagonal limit thus provides the kinematical mechanism through  which temperature enters the flat-space construction.} Further simplifications arise when additional relations between the conformal dimensions are imposed. For instance, by equating the external dimensions pairwise, $\Delta_2 = \Delta_3$ or $\Delta_1 = \Delta_4$, the two-term expression \eqref{block_bare} for the four-point $t$-channel conformal block reduces from Appell functions $F_4$ to a  single ${}_3 F_2$ hypergeometric function. The corresponding reduction formula derived in the present work, see \eqref{F4_diagonal}, renders the analytic structure of the diagonal blocks manifestly transparent. Importantly, the diagonal limit will also play a key role in our consideration of thermal and defect conformal blocks.

\section{Thermal conformal partial waves}
\label{sec:thermal}

Let us consider the scalar exchange in the thermal one-point function \eqref{th_corr_def}. Inserting the shadow projector \eqref{projector_def} into  \eqref{th_corr_def} we define the \tcpw  as
\be
\label{tCPW_trace}
\Tr_{\mathcal H}\Big[\Pi_{\Delta}\,\phi(x)\,q^D\Big]
=
C_{\widetilde{\Delta}\Delta_\phi\Delta}\,
\Upsilon_{\Delta}^{\Delta_\phi}(q,x)\,,
\ee
where the right-hand side admits an integral representation  \cite{Gobeil:2018fzy, Alkalaev:2023evp,Alkalaev:2024jxh}:
\be
\Upsilon_{\Delta}^{\Delta_\phi}(q,x)
=
q^\Delta
\int_{\mathbb{R}^d} {\rm d}^d x_0\,
V_{\widetilde{\Delta}\Delta_\phi\Delta}(x_0,x,qx_0)\,.
\ee
Substituting \eqref{corr_V} and rescaling the integration variable as $x_0\to x_0/q$ we  obtain the following expression,
\be
\label{tCPW_explicit}
\Upsilon_{\Delta}^{\Delta_\phi}(q,x)
=
\frac{q^{\widetilde{\Delta}}}{(1-q)^{\widetilde{\Delta}_\phi}}
\int_{\mathbb{R}^d} {\rm d}^d x_0\;
{X_{01}^{-c_1}X_{02}^{-c_2}(x_0^2)^{-c_3}}\,.
\ee
Here, we have effectively introduced two points $x_1=x$ and $x_2=qx$, and defined
\be
\label{t2_parameters}
\ba{c}
\dps
c_1=\frac{\Delta_\phi+\Delta-\widetilde{\Delta}}{2}\,,
\qquad
c_2=\frac{\widetilde{\Delta}+\Delta_\phi-\Delta}{2}\,,
\qquad
c_3=\frac{\widetilde{\Delta}_\phi}{2}\,,
\\[12pt]
c_1+c_2+2c_3=d\,,
\ea
\ee
where the linear constraint ensures the required  scaling properties of \eqref{tCPW_explicit}.

\subsection{Thermal conformal integrals}

Thermal conformal partial waves are expressed through the parametric two-point thermal conformal integral \cite{Alkalaev:2024jxh}:
\be
\label{t2}
T_2^{a_1,a_2,a_3}(x_1,x_2)
=
\int_{\mathbb{R}^d} \frac{{\rm d}^d x_0}{\pi^{\frac{d}{2}}}
\,X_{01}^{-a_1}X_{02}^{-a_2}(x_0^2)^{-a_3}\,.
\ee
Here, the  coordinates $x_1, x_2$ are  independent and the parameters $a_1, a_2, a_3$ are unconstrained. This integral is also known in the Feynman-integral literature as the triangle vertex integral. It has been explicitly  evaluated in \cite{Boos:1987bg, Boos:1990rg}.

Thermal conformal integrals are not independent, but rather emerge as a specific limit of the conformal integrals governing flat-space kinematics.  To make this relation explicit, we compare \eqref{t2} with \eqref{box_def} and find \cite{Alkalaev:2024jxh}:
\be
\label{t2_to_box}
T_2^{a_1,a_2,a_3}(x_1,x_2)
=
\lim_{\begin{subarray}{l}
\;x_3\to 0 \\ \;x_4\to\infty \end{subarray}}
\left(
x_4^{2a_4}\,
I_4^{a_1,a_2,a_3,a_4}(x_1,x_2,x_3,x_4)
\right)\Big|_{a_4=d-a_1-a_2-a_3}.
\ee

In \cite{Alkalaev:2024jxh}, we assumed the constraint $a_1+a_2+2a_3=d$, which is specific to the scalar sector, as seen, for example, from \eqref{t2_parameters}. However, even without such restrictions on the parameters, the integrals \eqref{t2} remain related to conformal integrals through \eqref{t2_to_box} and, therefore, inherit covariance under the residual conformal transformations
$o(d)\oplus o(1,1)$ that preserve the two fixed points. Such integrals arise, e.g., in the computation of spinning thermal conformal blocks
\cite{Alkalaev:workinprogress}. This motivates using the term \textit{thermal conformal} integral in a broader sense than in \cite{Alkalaev:2024jxh}.

This perspective  allows us to leverage known flat-space results in the thermal setting. In particular, the thermal integral $T_2$ inherits the four-term decomposition \eqref{box_sum}:
\be
\label{t2_sum}
T_2^{a_1,a_2,a_3}(x_1,x_2)
=
\mathcal T_2^{\langle 234\rangle}(\bm a|\bx)
+\mathcal T_2^{\langle 134\rangle}(\bm a|\bx)
+\mathcal T_2^{\langle 124\rangle}(\bm a|\bx)
+\mathcal T_2^{\langle 123\rangle}(\bm a|\bx)\,,
\ee
where
\be
\mathcal T_2^{\langle ijk\rangle}(\bm a|\bx)
\equiv
\lim_{\begin{subarray}{l}
\;x_3\to 0 \\ \;x_4\to\infty \end{subarray}}
\left(
x_4^{2a_4}\,
\Phi_4^{\langle ijk\rangle}(\bm a|\bx)
\right)\Big|_{a_4=d-a_1-a_2-a_3}\,.
\ee
Contrary to the flat-space case \eqref{box_basis_1}, taking the limit violates  the cyclic symmetry already in the defining relation \eqref{t2_to_box}, so the functions $\mathcal T_2^{\langle ijk\rangle}$ are no longer related by cyclic permutations. We conclude that the thermal conformal integral is expressed as a sum of four hypergeometric functions of two variables, thereby matching the results of \cite{Boos:1987bg, Boos:1990rg}.

We now consider the four-point conformal integral \eqref{t2_to_box} in the diagonal configuration \eqref{diag1}. If the parameters are additionally constrained by  $a_3 = a_4$, or equivalently $a_1+a_2+2a_3 = d$, the expression \eqref{t2_sum} is considerably simplified. In particular, all the Appell functions $F_4$ reduce to the generalized hypergeometric functions ${}_3F_2$ by virtue of  the reduction formula \eqref{F4_diagonal}. Under these conditions, the four terms in \eqref{t2_sum} combine into pairs:
\be
T_2^{a_1,a_2,a_3}(x_1,q x_1)\Big|_{a_1+a_2+2a_3=d}
=
\mathcal T_{+}^{\bm a}(q,x_1)+\mathcal T_{-}^{\bm a}(q,x_1)\,,
\ee
where
\be
\label{t2_B_S}
\ba{c}
\dps
\mathcal T_{+}^{\bm a}(q,x_1)
=
\mathcal T_2^{\langle 234\rangle}(\bm a|\bx)+\mathcal T_2^{\langle 123\rangle}(\bm a|\bx)\,,
\\[15pt]
\dps
\mathcal T_{-}^{\bm a}(q,x_1)
=
\mathcal T_2^{\langle 134\rangle}(\bm a|\bx)+\mathcal T_2^{\langle 124\rangle}(\bm a|\bx)\,.
\ea
\ee
These two functions,  expressed through the generalized hypergeometric function ${}_3F_2$, have distinct asymptotic behavior as $q\to 0$. It follows that they will contribute to either conformal $(+)$ or shadow $(-)$ thermal blocks.\footnote{It is worth noting that  in the  limit considered, a discrete symmetry emerges:
\be
\mathcal T_{-}^{\bm a}(q,x_1)
=
\sigma\circ
\mathcal T_{+}^{\bm a}(q,x_1)\,,
\ee
where $\sigma$ is the transposition acting on both coordinates and parameters as $(x_1,q, a_1)\leftrightarrow(x_1,1/q, a_2)$.} In other words, the thermal conformal and shadow blocks arise from the same set of flat-space basis functions, grouped according to their low-temperature asymptotics.

\subsection{From flat-space to thermal conformal blocks}

The \tcpw \eqref{tCPW_explicit}  is expressed in terms of the thermal conformal integral evaluated in the diagonal configuration discussed above. This observation allows for an explicit identification with the \fcpw \eqref{CPW_integral} taken in the same configuration, provided that the parameters of the corresponding conformal integrals are matched as
\be
\label{b=c}
b_1=c_1\,,
\qquad
b_2=c_2\,,
\qquad
b_3=c_3\,,
\qquad
b_4=c_3\,.
\ee
Such a  correspondence  is achieved by the specific choice of conformal dimensions \eqref{conformal_dimensions}.\footnote{The conditions \eqref{b=c} fix only the differences of conformal dimensions in the flat-space $\mathrm{CPW}$. This yields a number of equivalent parameterizations, among which  \eqref{conformal_dimensions} is a convenient choice.} The resulting identification reads
\be
\label{thermal_flat_cpw}
\Upsilon_{\Delta}^{\Delta_\phi}(q,x)
=
|x|^{2 \Delta - \widetilde\Delta_\phi}
\frac{q^{\Delta_\phi+\Delta}}{(1-q)^{\widetilde{\Delta}_\phi}}\,
\lim_{y\to\infty}\left[y^{2\widetilde{\Delta}_\phi}\,\tPsiDL\right].
\ee

Thus, the scalar \tcpw is obtained by considering  the $t$-channel \fcpw in a specific  kinematical regime characterized by a particular arrangement of points and conformal dimensions.\footnote{The residual conformal covariance of the \tcpw is inherited from the full conformal covariance of the {\tt f}CPW. Consequently, the required scaling behavior of the \tcpw\hspace{-1.5mm}, follows directly from the flat-space construction.} According to \eqref{CPW_sum}, the \tcpw decomposes into  conformal and shadow blocks as
\be
\label{tCPW_zero_combination}
\Upsilon_{\Delta}^{\Delta_\phi}(q,x)
=
K_{\widetilde{\Delta}}^{\Delta_\phi\Delta}\,
\mathcal F_{\Delta}^{\Delta_\phi}(q,x)
+
K_{\Delta}^{\widetilde{\Delta}_\phi\Delta}\,
\mathcal F_{\widetilde{\Delta}}^{\Delta_\phi}(q,x)\,,
\ee
where the two terms differing in their low-temperature $q\to 0$ asymptotics stem from \eqref{t2_B_S}. In particular, the thermal conformal  block is given by
\be
\mathcal F_{\Delta}^{\Delta_\phi}(q,x) = \Big(K_{\widetilde{\Delta}}^{\Delta_\phi\Delta}\,\Big)^{-1}\frac{q^{\widetilde{\Delta}}}{(1-q)^{\widetilde{\Delta}_\phi}}\,
  \mathcal T_{+}^{\bm c}(q,x)\,,
\ee
with the parameters $\bm c$  defined  in \eqref{t2_parameters}, cf. \eqref{block_sum}. Finally, the reduction formula \eqref{F4_diagonal} gives the following closed-form expression:
\be
\label{tF}
\mathcal F_{\Delta}^{\Delta_\phi}(q,x)
=
|x|^{-\Delta_\phi}\,
\frac{q^\Delta}{(1-q)^{2\Delta}}\,
{}_3F_2
\left[
\begin{array}{c}
\Delta-\frac{d}{2}+\frac{1}{2},\
\Delta-\frac{\Delta_\phi}{2},\
\Delta-\frac{\widetilde{\Delta}_\phi}{2}
\\[2mm]
\Delta,\
1+2\Delta-d
\end{array}
\Bigg|\,\frac{-4q}{(1-q)^2}
\right].
\ee
This expression satisfies the expected asymptotic behavior \eqref{tF_asymptot}, while the shadow block is generated via the substitution $\Delta\to\widetilde{\Delta}$. Importantly, the result \eqref{tF} agrees with the expression obtained in \cite{Gobeil:2018fzy} from a conjectured AdS integral representation. Since our derivation relies entirely on CFT methods, it provides an independent confirmation of that construction.

The thermal conformal block \eqref{tF} exhibits several notable properties. First, the combination $|x|^{\Delta_\phi}\mathcal F_{\Delta}^{\Delta_\phi}(q,x)$ is independent of the argument $x$ of the external primary operator and is symmetric under the reflection
\be
\label{th_block_symmetry}
\Delta_\phi \longleftrightarrow \widetilde{\Delta}_\phi\,.
\ee
Second, for $\Delta_\phi=0$  the  hypergeometric function ${}_3F_2$  reduces to  ${}_2F_1$, allowing for further simplification via the quadratic transformation \eqref{2F1_quadr_transform_2}:
\be
\label{tF_character_1}
\mathcal F_{\Delta}^{\Delta_\phi=0}(q,x)
=
\frac{q^\Delta}{(1-q)^{2\Delta}}\,
{}_2F_1
\left[
\begin{array}{c}
\Delta-\frac{d}{2}+\frac{1}{2},\
\Delta-\frac{d}{2}
\\[2mm]
1+2\Delta-d
\end{array}
\Bigg|\,\frac{-4q}{(1-q)^2}
\right] =
\frac{q^\Delta}{(1-q)^d}\;,
\ee
Thus, the thermal conformal block yields the character of the conformal algebra $o(d+1,1)$ \cite{Dolan:2005wy}. For $d=2$ the expression  \eqref{tF} reproduces the  global torus  block \cite{Hadasz:2009db}.

Finally, another distinguished point in the parameter space is reached when the exchanged conformal dimension is self-dual: $\Delta=\widetilde{\Delta}=d/2$. For concreteness, let us set $d=4$ and choose the external dimension to be self-dual as well, $\Delta_\phi=\widetilde{\Delta}_\phi=2$.\footnote{E.g., this corresponds to  $\varphi^2(x)$ in the $d=4$ free scalar theory.} Substituting these values into \eqref{t2_parameters}, one finds that all $c_i=1$. On the flat-space side, this corresponds to a configuration in which all external operators, as well as the exchanged operator, have conformal dimension two, see \eqref{conformal_dimensions}. The corresponding thermal integral is therefore obtained as the diagonal limit of the non-parametric four-point conformal integral, which admits a representation in terms of the Bloch--Wigner function, see \cite{Alkalaev:2024jxh} for details. The thermal conformal and shadow blocks in \eqref{tCPW_zero_combination} coincide and are given by
\be
\label{selfdual_thermal_block}
\mathcal{F}_{\Delta=2}^{\Delta_\phi=2}(q,x)
\Big|_{d=4} = |x|^{-2}\frac{q^2}{(1-q)^3}\,.
\ee
This function has been obtained directly by taking the self-dual limit of \eqref{tF}.
Computing the \tcpw in the self-dual limit is, however, a delicate matter because upon substituting $\dl = \frac{d}{2}$ the coefficients in the RHS of \eqref{tCPW_zero_combination}  diverge. Nevertheless, the full \tcpw remains finite and is given by a linear combination of the block and the derivative of the block with respect to the conformal dimension. A more detailed discussion of the \tcpw in the self-dual limit and its explicit computation for
$d=4,\, \dl_\phi  = 2$ are presented in Appendix \bref{app:self-dualtCPW}.

The discussion above clarifies the basic mechanism of our construction. \tcpws are formulated  within the shadow formalism and subsequently identified with a special configuration of  \fcpws in the  diagonal limit. In this way, the thermal conformal block is recovered from the $t$-channel conformal block \eqref{block_bare} and satisfies the asymptotic condition \eqref{tF_asymptot}.

\section{Defect conformal partial waves}
\label{sec:defect_scalar}

In this section, we establish a relation between \tcpws and \dcpws, in close analogy with the flat-space analysis discussed above. To this end, we adapt  the shadow formalism to the presence of defects (see also \cite{Fukuda:2017cup}).

Let us  consider a flat $p$-dimensional defect embedded in $\mathbb R^d$.\footnote{The techniques for constructing conformal blocks in  \dcft are thoroughly developed \cite{McAvity:1995zd, DeWolfe:2001pq, Liendo:2012hy, Billo:2016cpy, Gadde:2016fbj, Rastelli:2017ecj, Lauria:2017wav, Mazac:2018biw, Lauria:2018klo, Herzog:2020bqw, Kobayashi:2018okw, Liendo:2019jpu, Isachenkov:2018pef, Buric:2020zea}. For an introduction to \dcft\hspace{-1mm}, see e.g. \cite{Herzog:2021_defects_notes}.} Its presence breaks the conformal symmetry algebra $o(d+1,1)$ down to $o(p+1,1)\oplus o(d-p)$, which corresponds to conformal transformations along the defect and rotations in the transverse directions. Accordingly, local operators are partitioned into two classes. Defect operators are supported on the defect and transform in representations of the conformal algebra $o(p+1,1)$, thereby defining an ordinary CFT there. Bulk operators, inserted away from the defect, probe the coupling between the bulk and defect degrees of freedom, giving  rise to a richer set of observables. In what follows, we focus on bulk operators.

A characteristic feature of \dcfts is that one-point functions of bulk operators do not necessarily vanish. Instead, the broken conformal symmetry $o(p+1,1)\oplus o(d-p)$ fixes their form up to a model-dependent constant. For a scalar primary operator $\phi(x)$ of conformal dimension $\Delta_\phi$, the one-point function is given by \cite{McAvity:1995zd, DeWolfe:2001pq, Liendo:2012hy}:
\be
\big\langle \phi(x) \big\rangle_{\text{defect}}
\,=\,
\frac{a_{\Delta_\phi}}{|x_\perp|^{\Delta_\phi}\!\!}\;\,,
\ee
where $|x_\perp|$ denotes the transverse distance to the defect. The expectation value  $\langle \dots \rangle_{\text{defect}}$ is defined as $\langle \cD | \cR(\dots)\rvac$, where $\rvac$ is the $o(d+1,1)$ invariant vacuum, while $\langle \cD |$ is a defect state invariant under the broken conformal symmetry; $\cR$ denotes the radial ordering. The state $\langle \cD |$ encodes both the dynamical information and the boundary conditions imposed by the defect.

We now consider the two-point function of bulk scalar operators. Similar to the four-point functions in $\RR^d$, this correlator admits different OPE channels \cite{Billo:2016cpy, Gadde:2016fbj}. The bulk-channel expansion can be  obtained by performing the OPE between $\phi_{\Delta_1}$ and $\phi_{\Delta_2}$ over intermediate operators $\phi_{\Delta_k}$, leading  to
\be
\label{defect_2pt}
\big\langle \phi_{\Delta_1}(x_1)\phi_{\Delta_2}(x_2)\big\rangle_{\text{defect}}
=
\sum_{\Delta_k}
C_{\Delta_1\Delta_2\Delta_k}\,
a_{\Delta_k}\,
f_{\Delta_k}^{\Delta_1\Delta_2}(x_1,x_2)
+\text{spinning contributions} \,,
\ee
where $C_{\Delta_1\Delta_2\Delta_k}$ are the standard OPE coefficients and $a_{\Delta_k}$ are the one-point coefficients. The functions $f_{\Delta_k}^{\Delta_1\Delta_2}(x_1,x_2)$ are the bulk-channel conformal blocks, summing up the contributions of the entire conformal family of the intermediate operator $\phi _{\Delta _{k}}$ to the correlator. The bulk conformal block $f_{\Delta_k}^{\Delta_1\Delta_2}$ obeys the asymptotic  condition
\be
\label{defect_block_asymptot}
f_{\Delta_k}^{\Delta_1\Delta_2}(x_1,x_2)
\,\sim\,
X_{12}^{\frac{\Delta_k-\Delta_1-\Delta_2}{2}}\,
|x_{1,\perp}|^{-\Delta_k}\,,
\ee
dictated by the non-vanishing one-point function of the exchanged primary operator in the OPE regime $X_{12}\to 0$.

\subsection{Shadow representation}

In order to construct a \dcpw\hspace{-1mm}, we insert the shadow projector \eqref{projector_def} into the two-point function \eqref{defect_2pt}:\footnote{The relation between an operator of conformal dimension $\Delta_k$ and its shadow of dimension $\widetilde{\Delta}_k$ induces a corresponding relation between the one-point coefficients $a_{\Delta_k}$ and $a_{\widetilde{\Delta}_k}$. We do not include  their ratio into the definition of the \dcpw\hspace{-1mm}, cf. footnote \bref{fn:shadow_ratio}.}
\be
\label{defect_projector}
\big\langle \Pi_{\Delta_k} \phi_1(x_1)\phi_2(x_2)\big\rangle_{\text{defect}}
=
C_{\Delta_1\Delta_2\Delta_k}\,
a_{\widetilde{\Delta}_k}\,
\Theta_{\Delta_k}^{\Delta_1\Delta_2}(x_1,x_2)\,,
\ee
where
\be
\label{defect_CPW_V}
\Theta_{\Delta_k}^{\Delta_1\Delta_2}(x_1,x_2)
=
\int_{\mathbb R^d} {\rm d}^d x_0\;
V_{\Delta_1\Delta_2\Delta_k}(x_1,x_2,x_0)\,
\frac{1}{|x_{0,\perp}|^{\widetilde{\Delta}_k}\!\!}\;\,.
\ee
Substituting \eqref{corr_V} yields  the following integral representation:
\be
\label{defect_CPW_integral}
\Theta_{\Delta_k}^{\Delta_1\Delta_2}(x_1,x_2)
=
X_{12}^{\frac{\Delta_k-\Delta_1-\Delta_2}{2}}
\int_{\mathbb R^d} {\rm d}^d x_0\;
{X_{01}^{-d_1}X_{02}^{-d_2}(x_{0,\perp}^2)^{-d_3}
}\,,
\ee
where
\be
\ba{c}
\label{defect_parameters}
\dps
d_1=\frac{\Delta_k+\Delta_{12}}{2}\,,
\qquad
d_2=\frac{\Delta_k-\Delta_{12}}{2}\,,
\qquad
d_3=\frac{d-\Delta_k}{2}\,,
\\[12pt]
\dps
d_1+d_2+2d_3=d\,.
\ea
\ee
The \dcpw decomposes into conformal and shadow blocks, uniquely determined by their asymptotics as $X_{12}\to 0$. These asymptotics  are given by \eqref{defect_block_asymptot} for the conformal dimensions of the intermediate channel $\Delta_k$ and $\widetilde \Delta_k $, respectively.

Formula \eqref{defect_CPW_integral} provides a generalization of both the \fcpw \eqref{CPW_integral} and the \tcpw \eqref{tCPW_explicit} to the case of {{\tt d}CFT}. In particular, the corresponding ${\tt d}\mathrm{CPW}$ is given by a special type of integral, whose integrand is a product of power-law factors, one of which depends solely  on the transverse components of the integration point $x_{0,\perp}$. It is therefore of interest to relate this integral directly to the four-point conformal integral \eqref{box_def}, in close analogy with the relation \eqref{t2_to_box}.

\subsection{From defect to thermal conformal blocks:  scalar case}

In the special case of a point-like defect ($p=0$), we have $x_{0,\perp}=x_0$, so that the integral in \eqref{defect_CPW_integral} reduces exactly to the thermal conformal integral \eqref{t2}. The broken  conformal  algebra  then becomes   $o(1,1)\oplus o(d)$,   coinciding  with the symmetry preserved by thermal one-point functions, cf. \eqref{th_corr_def}. These observations allow for establishing  the relation between defect and thermal CPWs in two steps. First, the two  external points in the \dcpw are restricted to the diagonal configuration \eqref{diag2}, controlled by a real parameter $q$, which is naturally interpreted as the thermal one. Second, the parameters of the integrals \eqref{t2_parameters} and \eqref{defect_parameters} must be identified as
\be
\label{c=d}
c_1 = d_1\,,
\qquad
c_2 = d_2\,,
\qquad
c_3 = d_3\,.
\ee
This identification is achieved by choosing the conformal dimensions in \dcft as in \eqref{conformal_dimensions_defect}. Notably, this choice effectively swaps the external and exchanged operators between the two setups: the exchanged operator in \tcft becomes external in \dcft and vice versa. Under these conditions, the \dcpw reduces to the \tcpw as
\be
\label{thermal_defect_scalar}
\Upsilon_{\Delta}^{\Delta_\phi}(q,x)
=
|x|^{\,\widetilde{\Delta}_\phi}\,q^{\widetilde{\Delta}}\,
\Theta_{\Delta_\phi}^{\Delta\widetilde{\Delta}}(x,qx)\Big|_{p=0}\,.
\ee

The correspondence formula  \eqref{thermal_defect_scalar} does not extend straightforwardly to individual conformal blocks. This stems from the specific  assignment of conformal dimensions used above, which effectively interchanges the roles of external and exchanged operators.  In this sense, the thermal block emerges from a reshuffling of defect data rather than as the direct image of a single conformal block. Moreover, the bulk conformal block is identified by the bulk OPE limit \eqref{defect_block_asymptot},  corresponding to the high-temperature limit $q\to 1$ in the diagonal configuration. By contrast, the thermal conformal block is defined by the low-temperature limit $q\to 0$ \eqref{tF_asymptot}. Thus, reconstructing the thermal conformal block within the \dcft setup requires first computing the full \dcpw and analytically continuing it to the region relevant for small $q$.  Only then can the thermal conformal block be isolated via its asymptotic behavior as $q\to 0$. From the \dcft perspective, the thermal conformal block therefore receives contributions from both the bulk conformal block and its shadow counterpart.

Summarizing the previous sections, we have shown that the \tcpw can be obtained from both the \fcpw and the \dcpw by restricting them to special operator configurations. As a byproduct, this analysis also reveals an identification between the \dcpw and the \fcpw (see fig. \bref{fig:defect-flat}):
\be
\label{defect-flat}
\Theta_{\Delta_\phi}^{\Delta\widetilde{\Delta}}(x_1,x_2)\Big|_{p=0} = |x_1-x_2|^{-\widetilde{\Delta}_\phi} |x_2|^{2\Delta -\widetilde{\Delta}_\phi} \,
\lim_{y\to\infty}\left[y^{2\widetilde{\Delta}_\phi}\,
{}^{(t)}\Psi_{\Delta}^{\Delta,\Delta_\phi, \Delta,\widetilde\Delta_\phi}(x_1, x_2,0,y) \right].
\ee
In contrast to the analogous relations in the thermal setup \eqref{thermal_flat_cpw}, \eqref{thermal_defect_scalar}, this relation does not require imposing the diagonal limit and holds for arbitrary $x_1$ and $x_2$.

\section{Spinning conformal partial waves}
\label{sec:spin}

In what follows, we consider a totally symmetric traceless primary operator of spin $l$, which appears as an external operator in the thermal setup and as an intermediate operator in the defect setup.

\subsection{Defect CPWs with spinning exchange operators}

To study a spinning exchange, one defines the spinning shadow projector \cite{Dolan:2011dv,Simmons-Duffin:2012juh}:
\be
\label{projector_spin}
\Pi_{\Delta_\phi,l}
=
\int_{\mathbb{R}^d} {\rm d}^d x_0\;
\phi_{\mu(l)}(x_0)\ket{0}\bra{0}\widetilde{\phi}^{\mu(l)}(x_0)\,,
\ee
where $\phi_{\mu(l)}(x)\equiv\phi_{\mu_1\cdots\mu_l}(x)$ is a traceless tensor, $\delta^{\mu_1\mu_2}\phi_{\mu_1\mu_2\mu_3...\mu_l}(x)=0$. The shadow operator $\widetilde{\phi}^{\mu(l)}(x)$ has conformal dimension $\widetilde{\Delta}_\phi$ and transforms in the same spin-$l$ representation.

Inserting \eqref{projector_spin} into the two-point function of bulk scalar primaries,
\be
\label{defect_spin_projector}
\big\langle
\Pi_{\Delta_\phi,l}\,
\phi_1(x_1)
\phi_2(x_2)
\big\rangle_{\text{defect}}
=
C_{\Delta_1\Delta_2\Delta_\phi}^{(l)}\,
a_{\widetilde{\Delta}_\phi}^{(l)}\,
\Theta_{\Delta_\phi,l}^{\Delta_1\Delta_2}(x_1,x_2)\,,
\ee
defines the spinning bulk-channel ${\tt d}$CPW, cf. \eqref{defect_projector}. There are two types of tensor structures involved. First, the three-point function of two scalars and one spinning operator \cite{Fradkin:1996is,Osborn:1993cr}:
\be
\label{3pt_scalar_scalar_spin}
\big\langle
\phi_{\Delta_1}(x_1)\,
\phi_{\Delta_2}(x_2)\,
\phi_{\Delta_3}^{\mu(l)}(x_3)
\big\rangle
=
C_{\Delta_1\Delta_2\Delta_3}^{(l)}
\,V_{\Delta_1\Delta_2\Delta_3}(x_1,x_2,x_3)\,
\bigl(\widehat{Z}^{\mu}(x_3|x_1,x_2)\bigr)^l\,,
\ee
where $C_{\Delta_1\Delta_2\Delta_3}^{(l)}$ denotes  the $0-0-l$ OPE coefficients,\footnote{We identify $C_{\Delta_1\Delta_2\Delta_3} \equiv  C_{\Delta_1\Delta_2\Delta_3}^{(0)}$, see \eqref{CPW_projector}.} $V_{\Delta_1\Delta_2\Delta_3}$ is given by \eqref{corr_V}, and the unit vector $\widehat{Z}^\mu$ is defined via
\be
\label{Z_def}
Z^\mu(x_i|x_j,x_k)
=
\frac{(x_i-x_j)^\mu}{X_{ij}}
-
\frac{(x_i-x_k)^\mu}{X_{ik}}\,,
\qquad
\widehat{Z}^\mu = \frac{Z^\mu}{|Z|}\,.
\ee
Second, the one-point function of a bulk spin-$l$ primary operator in the presence of a $p$-dimensional defect is fixed by symmetry to be \cite{Billo:2016cpy}:
\be
\label{defect_one_spin}
\big\langle \phi^{\mu(l)}(x) \big\rangle_{\text{defect}}
=
a_{\Delta_\phi}^{(l)}\,
\frac{(\widehat{x}_\perp^{\mu})^l}{|x_\perp|^{\Delta_\phi}\!}\,,
\qquad
\widehat{x}_\perp^{\mu} = \frac{{x}_\perp^{\mu}}{|{x}_\perp|}\,.
\ee
This expression vanishes for odd spin $l$, reflecting invariance under $x_\perp^\mu \to -x_\perp^\mu$. In \eqref{3pt_scalar_scalar_spin} and  \eqref{defect_one_spin} we use the shorthand notation $(Z^\mu)^l = Z^{\mu_1}\cdots Z^{\mu_l} -{traces}$.

Combining these ingredients we obtain an integral representation of the spinning bulk-channel \dcpw :
\be
\label{defect_spin_integral}
\Theta_{\Delta_\phi,l}^{\Delta_1\Delta_2}(x_1,x_2)
=
\int_{\mathbb R^d} {\rm d}^d x_0\;
V_{\Delta_1\Delta_2\Delta_\phi}(x_1,x_2,x_0)
\left( \widehat{Z}_\mu(x_0|x_1,x_2)\right)^l
\frac{(\widehat{x}_{0,\perp}^\mu)^l }{|x_{0,\perp}|^{\widetilde{\Delta}_\phi}\!}\;\,.
\ee

\paragraph{Point-like defect.} Let us specialize to the case $p=0$, for which $x_{0,\perp} = x_0$. We recall that the contraction of  two symmetric traceless tensors $A$ and $B$ can be expressed as \cite{Dolan:2000ut, Dolan:2011dv}:
\be
\label{gegenbauer}
(\widehat{A}^\mu)^l (\widehat{B}_{\mu})^l
=
\widehat{C}_l^{(\frac d2-1)}(\widehat{A} \cdot \widehat{B})\,,
\qquad
\widehat{C}_l^{(\epsilon)}(t) = \frac{l!}{2^l (\epsilon)_l}\,
C_l^{(\epsilon)}(t)\,,
\ee
where $C_l^{(\epsilon)}(t)$ is the Gegenbauer polynomial, and the dot $\cdot$ denotes the standard Euclidean inner product. Applying this identity in \eqref{defect_spin_integral} we obtain
\be
\label{defect_spin_cpw_final}
\Theta_{\Delta_\phi,l}^{\Delta_1\Delta_2}(x_1,x_2)
=
X_{12}^{\frac{\Delta_\phi-\Delta_1-\Delta_2}{2}}
\int_{\mathbb{R}^d} {\rm d}^d x_0\;
\frac{
\widehat{C}_l^{(\frac d2 -1)}(t)
}{
X_{01}^{d_1}\,
X_{02}^{d_2}\,
(x_0^2)^{d_3}
}\;,
\ee
where the parameters $d_i$ are the same  as in the scalar case \eqref{defect_parameters} upon setting $\Delta_k=\Delta_\phi$; the argument of the Gegenbauer polynomial is given by
\be
\label{gegenbauer_t}
t
=
\frac{1}{2}
\sqrt{\frac{X_{01}X_{02}}{X_{12} x_0^2}}
\left(
\frac{x_0^2-x_1^2}{X_{01}}
-
\frac{x_0^2-x_2^2}{X_{02}}
\right).
\ee

Formula \eqref{defect_spin_cpw_final} casts  the \dcpw into a form directly analogous to the conformal integral representation of the \fcpw \cite{Dolan:2000ut, Dolan:2011dv}. The spin dependence is entirely encoded in the Gegenbauer polynomial, while the remaining part coincides with that of the scalar exchange.

\subsection{From defect to thermal conformal blocks:  spinning case}

Let us consider the conformal block decomposition of the thermal one-point function of a spin-$l$ primary operator, cf. \eqref{th_corr_CB_expansion}:
\be
\label{th_corr_CB_expansion_spin}
\Tr_{\mathcal H}
\Big[
\phi^{\mu(l)}(x)\,
q^D
\Big]
=
\sum_{\Delta}
C_{\Delta\Delta_\phi\Delta}^{(l)}\,
\mathcal F_{\Delta}^{\Delta_\phi,\mu(l)}(q,x)
+ \text{spinning contributions}\,,
\ee
where the OPE coefficients $C_{\Delta\Delta_\phi\Delta}^{(l)}$ are defined by  \eqref{3pt_scalar_scalar_spin}.\footnote{These coefficients vanish for odd $l$, reflecting the antisymmetry of the tensor structure \eqref{Z_def}. As a consequence, the one-point function of operators with odd spin $l$ receives no contributions from the scalar sector.} As in the scalar case, the thermal conformal block $\mathcal F_{\Delta}^{\Delta_\phi,\mu(l)}$ admits an expansion in the thermal parameter $q$, with coefficients determined by matrix elements of states in the corresponding conformal $o(d+1,1)$ module. Its leading behaviour in the low-temperature limit reads
\be
\label{tF_asymptot_spin}
\mathcal F_{\Delta}^{\Delta_\phi,\mu(l)}(q,x)
=
\frac{(\widehat{x}^{\mu})^l}{|x|^{\Delta_\phi}}\; q^{\Delta} \bigl( 1 + O(q) \bigr)\,.
\ee
Moreover, the $x$-dependence is contained in the prefactor, completely fixed by the residual symmetry $o(1,1)\oplus o(d)$, so that all nontrivial information is encoded in the dependence on $q$, cf. \eqref{tF_asymptot}. This makes it convenient to contract the thermal correlator with the traceless tensor $(\widehat{x}^{\mu})^l$. Inserting the scalar shadow projector \eqref{projector_def} into the one-point function \eqref{th_corr_CB_expansion_spin} then yields  the spinning \tcpw
\be
\label{tCPW_trace_spin}
(\widehat{x}^{\mu})^l\,
\Tr_{\mathcal H}
\Big[
\Pi_\Delta\,
\phi_{\mu(l)}(x)\,
q^D
\Big]
=
C_{\widetilde\Delta\,\Delta_\phi\,\Delta}^{(l)}
\,
\Upsilon_{\Delta}^{\Delta_\phi,l}(q,x)\,,
\ee
which possesses  the following integral representation:
\be
\label{thermal_spin_cpw}
\Upsilon_{\Delta}^{\Delta_\phi,l}(q,x)
=
q^\Delta
\int_{\mathbb R^d} {\rm d}^d x_0\;
V_{\widetilde\Delta\Delta_\phi\Delta}(x_0,x,qx_0)
\,
\bigl(\widehat{Z}_{\mu}(x|x_0, q x_0)\bigr)^l (\widehat{x}^{\mu})^l\,.
\ee
Using \eqref{gegenbauer} and \eqref{corr_V} as well as  rescaling the integration variable as $x_0 \to x_0/q$, we arrive at
\be
\label{tCPW_explicit_spin}
\Upsilon_{\Delta}^{\Delta_\phi,l}(q,x)
=
\frac{q^{\widetilde{\Delta}}}{(1-q)^{\widetilde{\Delta}_\phi}}
\int_{\mathbb{R}^d} {\rm d}^d x_0\;
\frac{
\widehat{C}_l^{(\frac d2 -1)}(t)}
{X_{01}^{c_1}\, X_{02}^{c_2}\, (x_0^2)^{c_3}}
\,,
\ee
where the parameters $c_{1,2,3}$ are the same as in the scalar case \eqref{t2_parameters}. The argument $t$ of the Gegenbauer polynomial coincides with \eqref{gegenbauer_t} upon setting $x_1=x$ and $x_2=qx$.

A direct comparison of the integral representations \eqref{tCPW_explicit_spin} and \eqref{defect_spin_cpw_final} shows that they coincide under the same parameter identification as in the scalar case: $c_1 = d_1$, $c_2 = d_2$, and $c_3 = d_3$. Explicitly, this relation takes the form
\be
\label{thermal_defect_spin}
\Upsilon_{\Delta}^{\Delta_\phi,l}(q,x)
=
q^{\widetilde{\Delta}}\,|x|^{\widetilde{\Delta}_\phi}\,
\Theta_{\Delta_\phi,l}^{\Delta\widetilde{\Delta}}(x,qx)\Big|_{p=0}\;.
\ee
Thus, the spinning one-point \tcpw is obtained from the spinning two-point \dcpw by restricting the latter to the diagonal configuration \eqref{diag2} and choosing the conformal dimensions as in \eqref{conformal_dimensions_defect}. This identity extends the thermal-defect correspondence to operators with non-zero spin.

\section{Thermal Casimir equations}
\label{sec:casimir}

Recall that the flat-space conformal blocks are uniquely determined as eigenfunctions of the quadratic Casimir operator, supplemented with appropriate asymptotic conditions \cite{Dolan:2003hv}. Under the correspondence discussed above, we require the behavior of this Casimir system in the diagonal limit.

However, expanding the quadratic Casimir equation in powers of the transverse coordinates near this configuration generically couples different orders of the expansion. Consequently, the quadratic operator alone does not yield a closed differential equation at the diagonal point. This difficulty is bypassed by supplementing the quadratic Casimir equation with an additional differential constraint, naturally provided by the fourth-order Casimir operator \cite{Hogervorst:2013kva}. The resulting system allows for a consistent reduction to the diagonal limit, yielding a closed equation for the reduced function. Below, we briefly adapt this procedure to our notation and conventions.

It is convenient to factor out the conformally covariant prefactor and introduce the dimensionless $t$-channel scalar \fcpw $\psi (q,\bar q)$ via
\be
\tPsi=
\frac{1}{X_{23}^{\frac{\Delta_2+\Delta_3}{2}}X_{14}^{\frac{\Delta_1+\Delta_4}{2}}}
\left(\frac{X_{24}}{X_{34}}\right)^{\frac{\Delta_{32}}{2}}
\left(\frac{X_{34}}{X_{13}}\right)^{\frac{\Delta_{14}}{2}}
\, \psi (q,\bar q)\,,
\ee
where the variables $q,\bar q$ are defined in \eqref{uv_q}, cf. \eqref{block}. Here and in what follows, we simplify the notation by suppressing the channel index  as well as the dependence on the  conformal dimensions. We  also introduce the differential operator
\be
\label{Casimir_operator}
{D}
=
D_q + D_{\bar q}
+ (d-2)\,\frac{q \bar q}{q - \bar q}
\Big((1-q)\partial_q - (1-\bar q)\partial_{\bar q}\Big)\,,
\ee
where
\be
\label{D_operator}
D_q
=
(1-q)\,q^2\,\partial_q^2
-
\Big(1+\frac{\Delta_{14}}{2}-\frac{\Delta_{32}}{2}\Big)q^2\,\partial_q
+
\frac{\Delta_{14}\Delta_{32}}{4}\,q
\quad \text{and} \quad
D_{\bar q} = D_{q}\Big|_{q\to \bar q}\,.
\ee
The Casimir equations then take the form \cite{Dolan:2000ut, Dolan:2003hv, Dolan:2011dv, Hogervorst:2013kva}:
\be
\label{Casimir_Flat}
\begin{aligned}
& \Big({D} -
\frac{1}{2}\,\Delta(\Delta-d)\Big)\,
\psi (q,\bar q) =0 \,,
\\[1mm]
& \Big(D_q - D_{\bar q}\Big)
\psi (q,\bar q)
=0\,,
\end{aligned}
\ee
where the first and second equations originate from the quadratic and quartic Casimir operators, respectively. For the scalar exchange considered here, the quartic Casimir operator, generally of fourth order, reduces to a second-order differential operator.  It is worth noting that in terms of the cross-ratios $u,v$ \eqref{uv_q} this  system is equivalent to the PDE system defining the Appell function $F_4$  \eqref{F4_equations}.

To take the diagonal limit, we expand the function $\psi (q,\bar q)$  around $q=\bar  q$. In terms of
\be
\eta=\frac{q-\bar q}{2}\,,
\ee
the diagonal limit corresponds to $\eta=0$. By construction, $\psi(q,\bar q)$ decomposes into  a sum of the (dimensionless) conformal and shadow blocks, which are symmetric under $q\leftrightarrow \bar q$, see \eqref{block_bare} and \eqref{uv_q}. As a consequence, the $\eta$-expansion contains only even powers of $\eta$:
\be
\psi(q,\bar q)
=
\psi_0(q)+\eta^2\psi_2(q)+\eta^4\psi_4(q)+\cdots\,.
\ee
Substituting this expansion into the Casimir equations \eqref{Casimir_Flat} and keeping the leading orders in $\eta$, we obtain a coupled  system of ODEs for $\psi_0(q)$ and $\psi_2(q)$. More precisely, the first and second equations \eqref{Casimir_Flat} yield, respectively,
\be
\label{diag_eq1_q}
\ba{c}
\dps
\,q^2\big(1-q\big)\,\psi_0''(q)
-\big(d+\Delta_{14}-\Delta_{32}\big)q^2\,\psi_0'(q)
+
\Big[\Delta_{14}\Delta_{32}\,q-\Delta(\Delta-d)\Big]\psi_0(q)
\\
[10pt]
\dps
+2(d-1)\,q^2(1-q)\,\psi_2(q)=0\,,
\ea
\ee
\be
\label{diag_eq2_q}
\ba{l}
\dps
\big(q-\frac32\,q^2\big) \psi_0''(q) -\left(2+\Delta_{14}-\Delta_{32}\right)q\,\psi_0'(q)
+\frac{\Delta_{14}\Delta_{32}}{2}\,\psi_0(q)
\\
[10pt]
\dps
\hspace{28mm}+ 2q^2\big(1-q\big)\,\psi_2'(q)
+\big(2q-(5+\Delta_{14}-\Delta_{32})q^2\big)\psi_2(q)=0\,.
\ea
\ee
These equations mix different orders of the $\eta$-expansion, reflecting the fact that the diagonal limit does not commute with the Casimir operators, thereby explaining why the quadratic Casimir equation alone is insufficient.

Solving \eqref{diag_eq1_q} algebraically for $\psi_2(q)$ and substituting the result into \eqref{diag_eq2_q} we obtain a  third-order ODE  for $\psi_0(q)$. Upon imposing the  parametrization of conformal dimensions \eqref{conformal_dimensions}, it takes the form
\be
\label{diag_ode_q}
\begin{aligned}
&2q^3(1-q)^2\,\psi_0'''(q)
+2q^2(1-q)\Big[(2d-4-3\Delta_\phi)q-d+2\Big]\psi_0''(q)
\\[1mm]
&\quad
+2q\Big[
\Big(-\Delta(\Delta-d)+(d-3\Delta_\phi-2)(d-\Delta_\phi-1)\Big)q^2
\\
&\qquad\qquad
+\Big(2\Delta(\Delta-d)-\Delta_\phi^2+3d\Delta_\phi-d^2-4\Delta_\phi+3d-2\Big)q
-\Delta(\Delta-d)
\Big]\psi_0'(q)
\\[1mm]
&\quad
+\Big[
-2(\Delta-\Delta_\phi)(d-\Delta-\Delta_\phi)(d-\Delta_\phi-1)\,q^2
\\
&\qquad
-(d-\dl_\phi)\Big(2\dl(\dl - d) + \dl_\phi (d-1)\Big)q
+2\Delta(\Delta-d)
\Big]\psi_0(q)=0\,.
\end{aligned}
\ee
This equation is equivalent to the defining ODE for the hypergeometric function ${}_3F_2$  \eqref{3F2_equation}. Near $q=0$ it possesses  three linearly independent solutions among which  we choose the one compatible with the low-temperature asymptotics of the thermal conformal block \eqref{tF_asymptot}:
\be
\label{diag_solution_q}
\psi_0(q)
=
|x|^{\Delta_\phi}\,(1-q)^{\widetilde{\Delta}_\phi}\,
\mathcal F_{\Delta}^{\Delta_\phi}(q,x)\,.
\ee
The right-hand side depends solely  on $q$, since the product $|x|^{\Delta_\phi} \mathcal F_{\Delta}^{\Delta_\phi}(q,x)$ is $x$-independent, cf. \eqref{tF}. Thus, in addition to both the shadow integral  and AdS integral representations, the thermal block is accessible  by solving the diagonally reduced Casimir equations.

\section{Conclusions}
\label{sec:conclusion}

In this paper, we have explored the interplay between conformal blocks on flat, thermal, and defect backgrounds, using the shadow formalism as a unifying framework. By establishing a precise correspondence between these settings, we have shown that one-point thermal conformal blocks can be systematically obtained from their four-point flat-space and two-point defect counterparts. In particular, thermal partial waves of symmetric traceless external operators can be reformulated as defect bulk-channel partial waves. In this sense, the defect CFT description provides a natural extension of the flat-space prescription beyond the scalar sector. Moreover, using this correspondence, we have derived the thermal Casimir equation from the underlying symmetry of the flat-space theory, without introducing chemical potentials.

Several directions deserve further investigation.

\begin{itemize}
\item A natural extension is to generalize our construction by including spinning exchanges in thermal correlators. In this case, both the thermal and defect CFT frameworks involve multiple independent tensor structures, and the contribution of a fixed exchanged spin generally decomposes into several conformal blocks with distinct OPE coefficients. It is therefore not a priori clear whether the correspondence between tensor structures, or between channels, remains one-to-one. Understanding this interplay in the presence of multiple structures is an interesting open problem. A related group-theoretic  issue is how saturating the unitarity bounds for external and intermediate operators affects the established correspondence.

\item The inclusion of  chemical potentials is expected to involve defects of higher dimension $p>0$. Moreover, a simple counting of invariant variables suggests that more general correlators should be considered on the defect CFT side (see \cite{Bianchi:2026orb} for a recent discussion on higher-point defect correlators). For instance, in $d=3$, the thermal one-point block with a single chemical potential depends on three invariant variables \cite{Gobeil:2018fzy}, whereas the bulk two-point function in a defect CFT depends on at most two cross-ratios \cite{Gadde:2016fbj, Herzog:2020bqw, Lauria:2017wav}. In this context, the recent work \cite{Drukker:2026nvv}, which introduces novel defect configurations, may provide a useful framework.

\item A further promising direction arises from the shadow representations \eqref{defect_spin_cpw_final} and \eqref{tCPW_explicit_spin}, in which the entire spin dependence is captured by a single Gegenbauer polynomial. On the one hand, expanding this polynomial in powers of its argument represents both defect and thermal CPWs as finite linear combinations of integrals, each of which can be evaluated in terms of hypergeometric functions. On the other hand, the standard flat-space analysis \cite{Dolan:2000ut,Dolan:2011dv} suggests that this is not the most fruitful approach. Instead, the recursion relations for Gegenbauer polynomials induce corresponding recursion relations for the conformal blocks. It would be interesting to analyze the diagonal limit from this perspective, as it may provide a link to the more systematic formalism of weight-shifting operators \cite{Costa:2011dw, Karateev:2017jgd, Karateev:2018oml}.

\item It would also be useful to analyze the bulk-channel Casimir equation in the defect setup. Although its derivation is expected to parallel the flat-space case \cite{Billo:2016cpy, Gadde:2016fbj}, its application to thermal blocks involves several technical issues. In particular, the cross-ratios used in our construction differ from the variables commonly employed in defect CFT. In addition, there is some freedom in choosing the prefactor in \eqref{block}, which may lead to different forms of the resulting differential equations. Clarifying these aspects and identifying the most convenient formulation is left for future work.

\item A related problem is to apply this formalism to higher-point thermal correlators on $S^1_\beta \times S^{d-1}$, which requires a more detailed analysis of the underlying conformal integrals. Beyond the one-point case, one encounters a richer structure of exchange channels \cite{Alkalaev:2023evp, Ammon:2025cdz}, making the identification of the corresponding blocks more challenging.\footnote{In flat-space CFT multipoint conformal blocks in the comb and other channels were extensively studied in \cite{Alkalaev:2015fbw,Rosenhaus:2018zqn,Fortin:2019zkm,Fortin:2020zxw,Fortin:2020yjz,Fortin:2020bfq, Poland:2021xjs, Ammon:2024axd, Buric:2020dyz, Buric:2021ywo, Buric:2021ttm, Buric:2021kgy}.}

\item A holographic interpretation of the proposed correspondence would be particularly interesting. While several works have already investigated aspects of thermal and defect AdS/CFT correspondence (see e.g. \cite{Karch:2000gx,DeWolfe:2001pq, Rastelli:2017ecj, Alday:2020eua, Parisini:2023nbd, Parisini:2022wkb, Karlsson:2021duj, Karlsson:2019qfi, Dodelson:2023vrw, Dodelson:2022yvn, Georgiou:2025mgg, Georgiou:2025wbg, Georgiou:2026zhf, Linardopoulos:2026mut, Gimenez-Grau:2023fcy, Grinberg:2020fdj,Rodriguez-Gomez:2021pfh,David:2022nfn, Chakravarty:2025ncy, Ceplak:2025dds, Afkhami-Jeddi:2025wra, Rodriguez-Gomez:2021mkk, Dodelson:2020lal, Dodelson:2023nnr, Ceplak:2024bja, Dodelson:2022eiz}), it remains important to systematically understand how the AdS bulk configurations associated with flat-space, thermal, and defect CPWs are related.

\item In \cite{Chen:2026ium}, conformal blocks in the presence of zero-dimensional defects were related to light-light-heavy-heavy four-point conformal blocks in flat-space CFT. This provides a close parallel to the correspondence discussed in the present paper, although the two identifications have different origins. In our case, the relation follows from an exact matching of CPWs on different backgrounds for specific configurations of operators. In \cite{Chen:2026ium}, instead, the defect description emerges in the limit of large conformal dimensions of the heavy operators. On the other hand, heavy-light correlators are known to be related to thermal correlators, see e.g. \cite{Fitzpatrick:2014vua, Fitzpatrick:2015zha, Fitzpatrick:2015foa, Fitzpatrick:2015dlt}. It would be interesting to understand whether these three perspectives can be combined into a unified framework.

\item Finally, a natural step is to study defects at finite temperature \cite{Barrat:2024aoa, Giombi:2026kdz}.

\end{itemize}


\noindent \textbf{Acknowledgements.} We are grateful to Nikita Misuna for fruitful discussions and collaboration on a related project.  Our work was supported by the Foundation for the Advancement of Theoretical Physics and Mathematics “BASIS”.

\appendix

\section{Generalized hypergeometric functions}
\label{app:functions}
In this appendix, we collect several standard formulas from \cite{Bateman:100233, AndrewsAskeyRoy1999} and then derive the relation \eqref{F4_diagonal}, which, to the best of our knowledge, is new.

\paragraph{Definitions.} The generalized hypergeometric function ${}_3F_2$ is defined by the power series
\be
\label{3F2_def}
{}_3F_2
\left[
\begin{array}{c}
\alpha_1,\ \alpha_2,\ \alpha_3 \\
\beta_1,\ \beta_2
\end{array}
\Bigg|\, z
\right]
=
\sum_{n=0}^{\infty}
\frac{(\alpha_1)_n(\alpha_2)_n(\alpha_3)_n}
{(\beta_1)_n(\beta_2)_n}\,
\frac{z^n}{n!}\,,
\ee
which converges for $|z|<1$.

The fourth Appell function is defined by the double power series
\be
\label{F4}
F_4
\left[
\begin{array}{c}
\alpha,\ \beta \\
\gamma,\ \gamma'
\end{array}
\Bigg|\, u,v
\right]
=
\sum_{m,n=0}^{\infty}
\frac{(\alpha)_{m+n}(\beta)_{m+n}}
{(\gamma)_m(\gamma')_n}\,
\frac{u^m v^n}{m!\,n!}\,,
\ee
which converges in the domain $\sqrt{|u|}+\sqrt{|v|}<1$.

We use the following notation for $\Gamma$-functions:
\be
\label{gammas}
\Gamma
\left[
\begin{array}{l l}
a_1, \ldots, a_n \\
b_1, \ldots, b_m
\end{array}
\right] = \frac{\Gamma(a_1, \ldots, a_n)}{\Gamma(b_1, \ldots, b_m)}\,, \qquad \Gamma(a_1, \ldots, a_n) = \prod_{i=1}^{n} \Gamma(a_i)\,.
\ee
\paragraph{Differential equations.} The function ${}_3F_2$ satisfies a third-order differential equation. In terms of the Euler operator $\theta = z\, d/dz$, it takes the form
\be
\label{3F2_equation}
\Big[
\theta(\theta+\beta_1-1)(\theta+\beta_2-1)
-
z(\theta+\alpha_1)(\theta+\alpha_2)(\theta+\alpha_3)
\Big]\,
{}_3F_2
\left[
\begin{array}{c}
\alpha_1,\ \alpha_2,\ \alpha_3 \\
\beta_1,\ \beta_2
\end{array}
\Bigg|\, z
\right]
=0\,.
\ee
Near $z=0$ this equation admits three linearly independent solutions (for generic values of the parameters).

The Appell function $F_4$ satisfies a system of two second-order differential equations. In terms of the Euler operators $\theta_u=u\, \partial/\partial u$ and $\theta_v=v\, \partial/\partial v$, this system takes the form
\be
\label{F4_equations}
\begin{aligned}
&\Big[
\theta_u(\theta_u+\gamma-1)
-u(\theta_u+\theta_v+\alpha)(\theta_u+\theta_v+\beta)
\Big]\,
F_4
\left[
\begin{array}{c}
\alpha,\ \beta \\
\gamma,\ \gamma'
\end{array}
\Bigg|\, u,v
\right]
=0\,,
\\[2mm]
&\Big[
\theta_v(\theta_v+\gamma'-1)
-v(\theta_u+\theta_v+\alpha)(\theta_u+\theta_v+\beta)
\Big]\,
F_4
\left[
\begin{array}{c}
\alpha,\ \beta \\
\gamma,\ \gamma'
\end{array}
\Bigg|\, u,v
\right]
=0\,.
\end{aligned}
\ee
Near $(u,v)=(0,0)$ this system admits four linearly independent solutions (for generic values of the parameters).

\paragraph{Transformation formulas.} There are two identities:
\be
\label{2F1_quadr_transform_1}
{}_2F_1
\left[
\begin{array}{c}
\alpha,\ \beta \\
1+\alpha-\beta
\end{array}
\Bigg|\, v
\right]
=
(1-\sqrt{v})^{-2\alpha}
{}_2F_1
\left[
\begin{array}{c}
\alpha,\ \alpha-\beta+\frac{1}{2} \\
1+2\alpha-2\beta
\end{array}
\Bigg|\, \frac{-4\sqrt{v}}{(1-\sqrt{v})^2}
\right],
\ee
\be
\label{2F1_quadr_transform_2}
{}_2F_1
\left[
\begin{array}{c}
\alpha,\ \beta \\
2\beta
\end{array}
\Bigg|\, \frac{4v}{(1+v)^2}
\right]
=
(1+v)^{2\alpha}\,
{}_2F_1
\left[
\begin{array}{c}
\alpha,\ \alpha-\beta+\frac12 \\
\beta+\frac12
\end{array}
\Bigg|\, v^2
\right].
\ee

\paragraph{Reduction formula.}
\begin{prop}
The following identity holds:
\be
\label{F4_diagonal}
\begin{array}{l}
\dps\hspace{-10mm}F_4
\left[
\begin{array}{c}
\alpha,\ \beta \\
\gamma,\ 1+\alpha-\beta
\end{array}
\Bigg|\, (1-\sqrt{v})^2, v
\right]
=
(1-\sqrt{v})^{-2\alpha}
\Gamma\left[
\begin{array}{c}
\gamma,\ \gamma-\alpha-\beta \\
\gamma-\alpha,\ \gamma-\beta
\end{array}
\right]
\\[6mm]
\dps
\hspace{40mm}\times \ {}_3F_2
\left[
\begin{array}{c}
\alpha,\ \alpha-\beta+\frac{1}{2},\ 1+\alpha-\gamma \\
1+2\alpha-2\beta,\ 1+\alpha+\beta-\gamma
\end{array}
\Bigg|\, \frac{-4\sqrt{v}}{(1-\sqrt{v})^2}
\right].
\end{array}
\ee
\end{prop}
\begin{proof}
To establish \eqref{F4_diagonal}, we first decompose the Appell function \eqref{F4} using the splitting relation
\be
\label{F4_split}
F_4
\left[
\begin{array}{c}
\alpha,\ \beta \\
\gamma,\ \gamma'
\end{array}
\Bigg|\, u,v
\right]
=
\sum_{m=0}^{\infty}
\frac{(\alpha)_m(\beta)_m}{(\gamma)_m}\,
\frac{u^m}{m!}\,
{}_2F_1
\left[
\begin{array}{c}
\alpha+m,\ \beta+m \\
\gamma'
\end{array}
\Bigg|\, v
\right].
\ee
After imposing the constraint on parameters $\gamma' = 1 + \alpha - \beta$, the quadratic transformation \eqref{2F1_quadr_transform_1} can be applied. Then, going to the diagonal limit $u=(1-\sqrt{v})^2$, we find
\be
\label{F4_reduction_step}
F_4
\left[
\begin{array}{c}
\alpha,\ \beta \\
\gamma,\ \gamma'
\end{array}
\Bigg|\, u,v
\right]
=
\sum_{n=0}^{\infty}
\frac{(\alpha)_n\left(\alpha-\beta+\frac{1}{2}\right)_n}
{\left(1+2\alpha-2\beta\right)_n}\,
\frac{1}{n!}
\left(\frac{-4\sqrt{v}}{(1-\sqrt{v})^2}\right)^n
{}_2F_1
\left[
\begin{array}{c}
\alpha+n,\ \beta \\
\gamma
\end{array}
\Bigg|\, 1
\right],
\ee
where we used the identity $(\alpha)_m(\alpha+m)_n=(\alpha)_n(\alpha+n)_m$. Finally, using the Gauss summation formula
\be
\label{Gauss_unit}
{}_2F_1
\left[
\begin{array}{c}
\alpha,\ \beta \\
\gamma
\end{array}
\Bigg|\, 1
\right]
=
\frac{\Gamma(\gamma)\Gamma(\gamma-\alpha-\beta)}
{\Gamma(\gamma-\alpha)\Gamma(\gamma-\beta)}\,,
\ee
we immediately arrive at \eqref{F4_diagonal}.
\end{proof}

\section{Self-dual \tcpw}
\label{app:self-dualtCPW}

In this section, we  show how the \tcpw in the self-dual limit is related to the conformal block and  present its explicit form for $d=4,\, \dl_\phi =2$. For arbitrary parameters, the \tcpw is given in \eqref{tCPW_zero_combination}. Note, however, that the decomposition into the block and the shadow block used there is not applicable to the self-dual case, since the coefficients $K_{\td}^{\dl_\phi \, \dl}$ and $K_{\dl}^{\dl_\phi \, \td}$ diverge under the substitution $\dl = \frac{d}{2}$.\footnote{This divergence is a consequence of the choice of normalization for the \tcpw\hspace{-1mm}. In \cite{Alkalaev:2024jxh}, a different normalization of the \tcpw was adopted, where  the ratio of structure constants was included in the definition of the \tcpw\hspace{-1mm}, see footnote \bref{fn:shadow_ratio}. Such a prefactor renders the coefficients in front of both the block and the shadow block finite in the self-dual limit, while causing the full \tcpw to vanish.} To circumvent this difficulty, we employ a regularization scheme $\dl = \frac{d}{2}+ \epsilon $ and expand the block and the coefficients in a Laurent series in $\epsilon$:
\be
\label{e-expansion_of_Block_and_K}
\begin{split}
&K_{\td}^{\dl_\phi \, \dl}\Big|_{\dl = \frac{d}{2} + \epsilon} = \frac{K_{-1}^{\dl_\phi}}{\epsilon} + K_0^{\dl_\phi} + O(\epsilon)\,,\\
&\mathcal{F}^{\dl_\phi }_{\dl}(q,x)\Big|_{\dl = \frac{d}{2} + \epsilon} = \mathcal{F}^{\dl_\phi}_{\frac{d}{2}}(q,x) + \epsilon \, \partial_\dl \mathcal{F}^{\dl_\phi}_{\dl}(q,x)\Big|_{\dl = \frac{d}{2}} + O(\epsilon^2)\,.
\end{split}
\ee
The expansions of $K_{\dl}^{\dl_\phi \, \td}$ and the shadow block can be obtained by the replacement $\epsilon \rightarrow -\epsilon$ in \eqref{e-expansion_of_Block_and_K}. Upon substituting these expansions into \eqref{tCPW_zero_combination}, the  poles in
$\epsilon$ cancel so that we can take the limit $\epsilon \to 0$:
\be
\label{Self-DualCPWDecomposition}
\lim_{\epsilon \to 0}\Upsilon_{\frac{d}{2} + \epsilon}^{\Delta_\phi}(q,x) = 2K_{-1}^{\dl_\phi} \,  \partial_\dl \mathcal{F}^{\dl_\phi}_{\dl}(q,x)\Big|_{\dl = \frac{d}{2}} + 2K_0^{\dl_\phi} \, \mathcal{F}^{\dl_\phi}_{\frac{d}{2}}(q,x)\,.
\ee
Thus, the \tcpw in the self-dual limit is represented as a combination of the block and its derivative with respect to the conformal dimension.

Now,  let us restrict ourselves to the case $d=4,\,  \dl_\phi = 2$. In this case, the conformal block is given by \eqref{selfdual_thermal_block} and the derivative  can be computed from the general expression \eqref{tF}:
\be
\label{DerivativeOfSelDualBlock}
\partial_\dl \mathcal{F}^{\dl_\phi = 2}_{\dl}(q,x)\Big|_{\dl = 2} =|x|^{-2} \frac{q}{(1-q)^3}\Big(q \ln q + (1-q)\ln(1-q)\Big) + |x|^{-2} \frac{q^2}{(1-q)^3}\,.
\ee
In our special case, $K_{0}^{\dl_\phi = 2} =-K_{-1}^{\dl_\phi = 2} = \pi^2 $. Therefore, upon substituting the block and its derivative into \eqref{Self-DualCPWDecomposition},  the last term in \eqref{DerivativeOfSelDualBlock} cancels with the block \eqref{selfdual_thermal_block}, and the \tcpw takes the form\footnote{A curious observation here is that, up to a power-law prefactor, the \tcpw is proportional to the Shannon binary entropy function, $S(q)=-q \ln q - (1-q)\ln(1-q)$, where $q\in (0,1)$.}
\be
\Upsilon_{\dl = 2}^{\Delta_\phi = 2}(q,x)\Big|_{d=4}  = -2 \pi^2 |x|^{-2}\frac{q}{(1-q)^3}\Big(q \ln q + (1-q)\ln(1-q)\Big)\,.
\ee
Note that although this \tcpw is by construction a combination of the block and its derivative, the term proportional to the block itself  is absent. This cancellation is not a general property of self-dual \tcpws\hspace{-1mm}, rather, it  arises here due to the specific choice of parameters $d$ and $\dl_\phi$. For generic self-dual dimensions, however, the block can be extracted from the \tcpw using  \eqref{Self-DualCPWDecomposition}.

\providecommand{\href}[2]{#2}\begingroup\raggedright\endgroup

\end{document}